\documentclass{article}

\usepackage{PRIMEarxiv}

\usepackage[utf8]{inputenc} % allow utf-8 input
\usepackage[T1]{fontenc}    % use 8-bit T1 fonts
\usepackage{hyperref}       % hyperlinks
\usepackage{url}            % simple URL typesetting
\usepackage{booktabs}       % professional-quality tables
\usepackage{amsmath}
\usepackage{amsthm}
\usepackage{amssymb}
\usepackage{amsfonts}
\usepackage{mathtools}
\usepackage{amsfonts}       % blackboard math symbols
\usepackage{nicefrac}       % compact symbols for 1/2, etc.
\usepackage{microtype}      % microtypography
\usepackage{fancyhdr}       % header
\usepackage{graphicx}       % graphics
\usepackage{times}
\usepackage{subfig}
\usepackage{xcolor}
\usepackage{array}
\usepackage{cite}
\usepackage{multirow}
\usepackage{adjustbox}
\usepackage{fontawesome}
\usepackage[font=footnotesize,labelfont=bf]{caption}
\usepackage{tikz-network}
\tikzstyle{block} = [draw, fill=blue!10, rectangle, 
    minimum height=1cm, minimum width=1cm]
\tikzstyle{sum} = [draw, fill=black!20, circle, node distance=1.5cm]
\tikzstyle{input} = [coordinate]
\tikzstyle{output} = [coordinate]
\tikzstyle{pinstyle} = [pin edge={to-,thin,black}]
\usepackage{enumitem}

\DeclareMathOperator{\diag}{diag}
\DeclareMathOperator{\tr}{tr}
\DeclareMathOperator{\vect}{vec}

\newtheorem{theorem}{Theorem}
\newtheorem{lemma}{Lemma}

\newtheorem{definition}{Definition}

\newtheorem{remark}{Remark}

\title{A Passivity-Agnostic Framework for \\ Distributed Adaptive Synchronization under \\ Unknown Leader Dynamics
\thanks{This material is based upon work supported in part by the U.S. Office of Naval Research under Grant Award N00014-21-1-2431; and in part by the U.S. National Science Foundation under Grant Award 2208182.}
}

\author{
  Moh Kamalul Wafi, Milad Siami \\
  Department of Electrical and Computer Engineering \\
  Northeastern University \\
  Boston, MA, USA\\
  \texttt{\{wafi.m, m.siami\}@northeastern.edu} \\
}

\begin{document}
\maketitle

\begin{abstract}
We present a passivity-agnostic framework for distributed adaptive synchronization under position-only communication, bounded disturbances, and unknown leader dynamics. By passivity-agnostic we mean the design does not require the closed loop system to be strictly positive real (SPR) a priori: it certifies SPR when present and recovers it by frequency shaping when absent. Followers are heterogeneous second-order systems with unknown (possibly unstable) dynamics. In the SPR regime, a structured reparameterization yields gradient-based adaptive error dynamics; Lyapunov analysis guarantees global asymptotic synchronization in the disturbance-free case, exact rejection of constant disturbances, and bounded responses to time-varying disturbances, with parameter convergence under persistent excitation. In the non-SPR regime, frequency shaping recovers effective passivity of the unshaped transfer function, enabling the same stability guarantees via standard passivity/Lyapunov arguments using Meyer-Kalman-Yakubovich (MKY) Lemma. Simulations across star, cyclic, path, and arbitrary graphs demonstrate scalable synchronization, robust tracking, and parameter adaptation under multiple disturbance profiles, confirming that the frequency-shaped non-SPR designs match the performance of the SPR case.
\end{abstract}
\allowdisplaybreaks

\section{Introduction}\label{sec:Introduction}
Distributed cooperative control of multi-agent systems (MAS) underpins consensus, synchronization, and formation in robotics, autonomous vehicles, smart grids, and sensor networks \cite{R1}. Beyond the classical assumption of identical agent dynamics \cite{R2,R3}, practical systems are heterogeneous \cite{R3,Wafi-IFACWC,R5}. Recent work addresses this heterogeneity under varied tasks and constraints: time-varying output formation with a nonautonomous leader of unknown input \cite{R6}, robust distributed PID consensus via weighted edge dynamics \cite{Wafi-QuadrupleTank}, optimization-based formation with dynamic event-triggered communication \cite{R8}, robustness to uncertain communication links \cite{R9}, and network-adapted guaranteed-cost designs \cite{R10}. These directions span realistic communication limits and architectures. However, the multi-agent dynamics are rarely known exactly in practice as parameters drift and unmodeled effects persist.

This motivates addressing uncertain dynamics with adaptive, scalable designs. Within this setting, many distributed adaptive methods \cite{R11,R12,R13,R14,R15,R16,R17}, explicitly or implicitly, rely on a strictly positive real (SPR) error channel because SPR implies passivity. By the Meyer–Kalman–Yakubovich (MKY) lemma, a quadratic storage function exists that yields a clean dissipation inequality. Practically, this lets one build a Lyapunov function in which the plant–controller cross terms cancel with a gradient update, ensuring bounded parameter estimates and convergence of the synchronization error. Related structural stability mechanisms based on positivity and Lur’e systems have also been studied in \cite{Hedesh-CDC,Hedesh-CCTA}.

In networked settings, however, SPR often fails—e.g., relative-degree/phase lags from limited communication exchange, unmodeled actuator/sensor effects, sampling or delays, and graph-induced lightly damped modes \cite{R24,R25}. When the channel is non-SPR, standard adaptive arguments typically require strong PE, leakage ($\sigma$-modification), or extra damping that may conflict with decentralization \cite{R18}. Existing compensator-based work on non-SPR is largely single-agent \cite{R20,R21,R22}; in networked settings, phase-lead/lag designs typically assume known dynamics \cite{R23}. A unified adaptive treatment that accommodates both SPR and non-SPR regimes for heterogeneous networks under limited communication and unknown dynamics remains largely open.

This paper develops a passivity-agnostic framework for distributed adaptive synchronization under position-only communication, bounded disturbances, and unknown leader dynamics. This passivity-agnostic design does not require a priori SPR of the closed-loop error channel. Furthermore, our main contributions are the followings.
\begin{enumerate}[leftmargin=*]
\item \textbf{SPR reparameterization.} A structured reparameterization in SPR design that yields gradient-type adaptive error dynamics (as opposed to \cite{R11,R12,R13,R14,R15,R16,R17}); Lyapunov analysis gives global asymptotic synchronization (disturbance-free), exact rejection of constant disturbances, bounded responses to time-varying disturbances, and parameter convergence under persistent excitation.
\item \textbf{Non-SPR shaping.} A frequency-shaping design that recovers effective passivity of the error channel, enabling the same stability guarantees via the MKY connection.
\item \textbf{Decentralized architecture.} No leader dynamics; compatible with heterogeneous (possibly unstable) second-order followers under position-only exchanges.
\item \textbf{Scalable validation.} Comprehensive simulations on star, cyclic, path, and arbitrary graphs show robust tracking and parameter adaptation under multiple disturbance profiles; runtime benchmarking confirms favorable scaling as followers grow, alongside a complexity analysis.
\end{enumerate}
In conclusion, compared with SPR-based designs for known dynamics \cite{R2,R3,Wafi-IFACWC,R5,R6,Wafi-QuadrupleTank,R8,R9,R10}, SPR-based adaptive methods for unknown dynamics \cite{R11,R12,R13,R14,R15,R16,R17}, and non-SPR approaches \cite{R20,R21,R22,R23}, which address the two regimes in isolation, the proposed passivity-agnostic framework handles both within one architecture, without a priori SPR certification and without a leader model, yielding a unified, position-only solution for synchronization in heterogeneous networks.

\textbf{Notation.}
For $p\in\mathbb{N}$, $\mathbb{R}^p$ denotes the $p$-dimensional Euclidean space and $I_p\in\mathbb{R}^{p\times p}$ the identity. The term $\diag\{A_i\}$ is the block-diagonal matrix with diagonal blocks $A_i$. The all-ones and all-zeros vectors in $\mathbb{R}^p$ are $\mathbf{1}_p=[1,\dots,1]^\top$ and $\mathbf{0}_p=[0,\dots,0]^\top$. The Kronecker product is $A\otimes B$. We use $\operatorname{tr}(A)$ for the trace, $\sigma(A)$ for the spectrum (eigenvalues) of $A$, and $\operatorname{vec}(A)$ for vectorization. Vector inequalities $>,\ge,<,\le$ are understood elementwise. For any square (not necessarily symmetric) matrix $A$, $A\succ0$ $(A\succeq 0)$ means $\Re\{\sigma_i(A)\}>0$ (respectively $\Re\{\sigma_i(A)\}\ge0$) for all eigenvalues $\sigma_i(A)$; for symmetric $A$ this coincides with positive (semi)definiteness.

% \textbf{Notations}. 
% The set $\mathbb{R}^p$ is the $p$-dimensional Euclidean space. The identity matrix of size $\mathbb{R}^{p\times p}$ is denoted as $I_p$. A diagonal block matrix with entries $A_i, \forall i$, is expressed as $\mathbf{A}=\diag\{A_i\}$. The vector of all ones is denoted by $\mathbf{1}_p=[1,\dots,1]^\top$ in $\mathbb{R}^p$, and similarly, for the vector of all zeros $\mathbf{0}_p=[0,\dots,0]^\top$ in $\mathbb{R}^p$. Also, $A\otimes B$ defines the Kronecker product between $A$ and $B$, while $\tr[A]$ and $\sigma(A)$ are the trace and the eigenvalues of a matrix $A$ in turn. The vectorization operator, $\vect(A)$, transforms matrix $A$ into a column vector. The integral term $\int_{t_1}^{t_n} \omega(\tau)\;d\tau$ is re-written as $\mathbb{I}[t_1,t_n]\{\omega\}$ for simplicity. In vectors, the symbols $>,\geq, <,\leq$ show element-wise inequalities. Additionally, $\mathbf{W}_j$ and $\mathbf{C}_j$ with $j \in \{u,v,w\}$ represent transfer functions.

\section{Communication Network}\label{sec:ComNetwork}
We model a network of $m + 1$ agents, with $m$ followers and a leader denoted as agent $0$, by a directed graph $\mathcal{G}=\{\mathcal{V},\mathcal{E},\mathcal{W}\}$ with vertices $\mathcal{V}=\{0,1,\dots,m\}$, edges $\mathcal{E} \subseteq \mathcal{V} \times \mathcal{V}$, and weights $\mathcal{W}=\{w_{ij}\}$. An edge $(i,j)\in\mathcal{E}$ means agent $j$ communicates to $i$ (not necessarily bidirectional). The weights quantify the interaction strength between agents and satisfy $w_{ij}>0$ if $(i,j)\in\mathcal{E}$ and $w_{ij}=0$ otherwise. The in-neighborhood (the set of incoming edges) of agent $i$ is defined as $\mathcal{N}_i:=\{\,j\mid (i,j)\in\mathcal{E}\,\}$. Furthermore, for analysis, we decompose $\mathcal{G}$ into two induced subgraphs: the follower-only subgraph $\mathcal{G}_m=\{\mathcal{V}_m,\mathcal{E}_m,\mathcal{W}_m\}$ with $\mathcal{V}_m=\{1,\dots,m\}$ and the leader–follower subgraph $\mathcal{G}_0=\{\mathcal{V}_0,\mathcal{E}_0,\mathcal{W}_0\}$ with $\mathcal{V}_0=\{0\}\cup\{i:(i,0)\in\mathcal{E}_0\}$, capturing edges from the leader to followers (see Fig.~\ref{Fig:network}).

In the subgraph $\mathcal{G}_m$, we define the degree matrix $\mathbb{D}_m\coloneqq\diag\{d_1,\dots,d_m\}$, where each entry $d_i$ represents the total weight of the neighbors of follower $i$, i.e., $d_i = \sum_{j:\,(i,j)\in\mathcal{E}_m} w_{ij}$. The adjacency matrix $\mathbb{A}_m$ is denoted as $\mathbb{A}_m(i,j)=w_{ij}$ if $(i,j)\in\mathcal{E}_m$ and $0$ otherwise, while the Laplacian matrix is expressed as $\mathbb{L}_m = \mathbb{D}_m - \mathbb{A}_m$. In the subgraph $\mathcal{G}_0$, the leader-weight matrix is diagonal and is defined as $\mathbb{A}_0(i,i) = w_{i0} $ if $(i, 0) \in \mathcal{E}_0$, and $0$ otherwise. This matrix governs the influence of the leader on each follower, playing a role in defining the network's connectivity.

Let $\mathbb{W}\coloneqq\diag\{w_1,\dots,w_m\}=\mathbb{D}_m+\mathbb{A}_0$, with $w_i=d_i+w_{i0}$. We adopt the balance condition $\mathbb{W}=I_m$, under which
\begin{align*}
    \mathbb{L} \coloneqq \mathbb{L}_m + \mathbb{A}_0 = \mathbb{W} - \mathbb{A}_m,
\end{align*}
satisfies
\begin{align}
    (\mathbb{L} - \mathbb{A}_0) \mathbf{1}_m = \mathbf{0}_m. \label{Bal}
\end{align}
Therefore, $(\mathbb{A}_m + \mathbb{A}_0) \mathbf{1}_m = \mathbf{1}_m$. If initially $\mathbb{W}\neq I_m$, we apply a normalization\footnote{Set $\tilde w_{ij}=w_{ij}/w_i$ for $(i,j)\in\mathcal{E}_m$ and $\tilde w_{i0}=w_{i0}/w_i$ with $w_i=d_i+w_{i0}$; then $\tilde{\mathbb W}=I_m$ and $(\tilde{\mathbb L}-\tilde{\mathbb A}_0)\mathbf{1}_m=\mathbf{0}_m$.} of edge weights (without altering connectivity) to enforce $\mathbb{W}=I_m$ and thus \eqref{Bal}. We formalize the network topology assumptions as follows.

\begin{remark}[Network Connectivity Conditions]\label{Rem:Threshold}
    The leader-weight matrix $\mathbb{A}_0$ is positive semi-definite, i.e., $\mathbb{A}_0 \succeq 0$, with entries $0\le w_{i0}\le 1$ for all $i\in\mathcal{V}_m$, and $\exists w_{i0} \neq 0$ for at least one follower with $(i,0)\in\mathcal{E}_0$. The leader is reachable (there exists a directed path from the leader to every follower), so that $\mathbb{L}=\mathbb{L}_m+\mathbb{A}_0$ is \emph{positive stable}, i.e., all eigenvalues satisfy $\Re\{\sigma_i(\mathbb{L})\}>0$.
\end{remark}
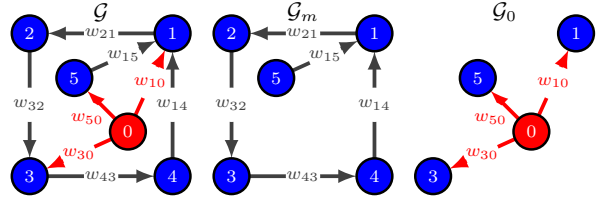
\begin{figure}[t!]
    \centering
    \scalebox{0.95}{{\begin{tikzpicture}
            \centering
            % \draw[black] (-1.5,-1.5) rectangle (1.5,1.5);
            \Text[x=0,y=1.3,fontsize=\small]{$\mathcal{G}$};
            \Vertex[x=.375,y=-.375,label=$0$,color=red,fontcolor=white,size=.5]{L}
            \Vertex[x=1,y=1,label=$1$,color=blue,fontcolor=white,size=.5]{1}
            \Vertex[x=-1,y=1,label=$2$,color=blue,fontcolor=white,size=.5]{2}
            \Vertex[x=-1,y=-1,label=$3$,color=blue,fontcolor=white,size=.5]{3}
            \Vertex[x=1,y=-1,label=$4$,color=blue,fontcolor=white,size=.5]{4}
            \Vertex[x=-.375,y=.375,label=$5$,color=blue,fontcolor=white,size=.5]{5}
            \Edge[Direct,color=red,label=$w_{10}$](L)(1)
            \Edge[Direct,color=red,label=$w_{30}$](L)(3)
            \Edge[Direct,color=red,position={below left},label=$w_{50}$](L)(5)
            \Edge[Direct,label=$w_{21}$](1)(2)
            \Edge[Direct,label=$w_{32}$](2)(3)
            \Edge[Direct,label=$w_{43}$](3)(4)
            \Edge[Direct,label=$w_{14}$](4)(1)
            \Edge[Direct,label=$w_{15}$](5)(1)
        \end{tikzpicture}}}
    \;
    \scalebox{0.95}{{\begin{tikzpicture}
            \centering
            % \draw[black] (-1.5,-1.5) rectangle (1.5,1.5);
            \Text[x=0,y=1.3,fontsize=\small]{$\mathcal{G}_m$};
            \Vertex[x=1,y=1,label=$1$,color=blue,fontcolor=white,size=.5]{1}
            \Vertex[x=-1,y=1,label=$2$,color=blue,fontcolor=white,size=.5]{2}
            \Vertex[x=-1,y=-1,label=$3$,color=blue,fontcolor=white,size=.5]{3}
            \Vertex[x=1,y=-1,label=$4$,color=blue,fontcolor=white,size=.5]{4}
            \Vertex[x=-.375,y=.375,label=$5$,color=blue,fontcolor=white,size=.5]{5}
            \Edge[Direct,label=$w_{21}$](1)(2)
            \Edge[Direct,label=$w_{32}$](2)(3)
            \Edge[Direct,label=$w_{43}$](3)(4)
            \Edge[Direct,label=$w_{14}$](4)(1)
            \Edge[Direct,label=$w_{15}$](5)(1)
        \end{tikzpicture}}}
    \;
    \scalebox{0.95}{{\begin{tikzpicture}
            \centering
            % \draw[black] (-1.5,-1.5) rectangle (1.5,1.5);
            \Text[x=0,y=1.3,fontsize=\small]{$\mathcal{G}_0$};
            \Vertex[x=.375,y=-.375,label=$0$,color=red,fontcolor=white,size=.5]{L}
            \Vertex[x=1,y=1,label=$1$,color=blue,fontcolor=white,size=.5]{1}
            \Vertex[x=-1,y=-1,label=$3$,color=blue,fontcolor=white,size=.5]{3}
            \Vertex[x=-.375,y=.375,label=$5$,color=blue,fontcolor=white,size=.5]{5}
            \Edge[Direct,color=red,label=$w_{10}$](L)(1)
            \Edge[Direct,color=red,label=$w_{30}$](L)(3)
            \Edge[Direct,color=red,position={below left},label=$w_{50}$](L)(5)
        \end{tikzpicture}}}
    \caption{An example of the graph $\mathcal{G}$ and two induced graphs $\mathcal{G}_m$ and $\mathcal{G}_0$ with $m = 5$, illustrating the decoupling of $\mathcal{G}$ and the assignment of $w_{ij}$—not actual communication.}
    \label{Fig:network}
\end{figure}

% Before proceeding, we clarify the leader’s role and the communication model. The leader’s dynamics are unknown and need not be identified; only the leader’s state information is disseminated through the network in a distributed manner, while followers exchange positions locally. When additional measurements (e.g., velocities) are available, they are used by the agents themselves without changing the distributed architecture. The synchronization objective is for all followers to track the leader’s trajectory and thus achieve synchronization over any communication network $\mathcal{G}$ satisfying the balance and reachability conditions in Remark~\ref{Rem:Threshold}. In the presence of unknown follower dynamics and time-varying disturbances, we employ passivity-agnostic distributed adaptive laws; the next two sections discuss the SPR (Sections~\ref{sec:SPR}) and non-SPR regimes (\ref{sec:nonSPR}), respectively.

In the following remark, we clarify the information flow.
\begin{remark}[Information Flow]\label{Rem:Communication}
    \textbf{Leader autonomy:} Only the leader’s state is broadcast through the network; its dynamics are unknown and not identified.
    \textbf{Information exchange}: Followers exchange positions with local neighbors only, as defined by the network topology $\mathcal{G}$.
    \textbf{Local measurements:} When additional measurements (e.g., velocities) are available to individual agents, they are utilized locally without being transmitted, preserving the distributed architecture.
\end{remark}

Finally, for the control objective, all followers must track the leader's trajectory to achieve network synchronization, provided $\mathcal{G}$ satisfies the connectivity conditions in Remark~\ref{Rem:Threshold}. To handle unknown follower dynamics and time-varying disturbances, passivity-agnostic distributed adaptive laws are employed; the next two sections discuss the SPR regime (Section~\ref{sec:SPR}) and the non-SPR regime (Section~\ref{sec:nonSPR}).

\section{Passivity-Certified (SPR) Synchronization}\label{sec:SPR}

\subsection{Theoretical Concepts}\label{subsec:SPR}
For clarity of exposition, we develop the results on second-order followers for the passivity-certified and the passivity-recovered construction. The dynamics are given by
\begin{align}
    J\ddot{\bar{x}} + B\dot{\bar{x}} = \bar{u}, \qquad \ddot{\bar{x}}, \dot{\bar{x}}, \bar{x}, \bar{u} \in\mathbb{R}^m\label{das:SPR:dynamics}
\end{align}
where $\bar{x}=[x_1,\dots,x_m]^\top$ and $\dot{\bar{x}}=[\dot{x}_1,\dots,\dot{x}_m]^\top$ represent the vectors of follower position and velocity, respectively. The diagonal matrices $J=\diag\{J_1,\dots,J_m\}$ with $J_i>0,\forall i$ and $B=\diag\{b_1,\dots,b_m\}$ are unknown constants related to inertia and damping coefficients, while $\bar{u}=[u_1,\dots,u_m]^\top$ defines the control input vector. Although both $\bar{x}$ and $\dot{\bar{x}}$ are measurable locally, only \emph{positions} are exchanged across the network. The transfer function of \eqref{das:SPR:dynamics} is therefore denoted by $\mathbf{M} = \diag\{M_1, \dots, M_m\}$.
Given the communication structure encoded by $\mathbb{A}_m$ and $\mathbb{A}_0$, and letting $\bar{x}_0=\mathbf{1}_m\otimes x_0$ be the replicated leader position, the total information received (incoming arrows) by follower $i$, for all $i \in \mathcal{V}_m$, from its neighbors $j \in \mathcal{N}_i$, is given by
\begin{align}
    \Bar{z} \coloneqq \mathbb{A}_m\Bar{x} + \mathbb{A}_0\Bar{x}_{0}, \label{das:SPR:z}
\end{align}
where $\bar{z}=[z_1,\dots,z_m]^\top$. The corresponding synchronization error (with position feedback $\bar{x}$) is therefore
\begin{align}
    \bar{e} = [e_1,\dots,e_m]^\top = -\mathbb{L}\bar{x} + \mathbb{A}_0\bar{x}_0 \eqqcolon \bar{z}-\bar{x}. \label{das:SPR:error}
\end{align}
We employ the following distributed control law:
\begin{align}
    \bar{u} \coloneqq c_u(\bar{e}) = 2\Phi\Lambda\Bar{e} + \Phi\Lambda^2\mathbb{I}\{\Bar{e}\} + \Phi\dot{\Bar{e}}, \label{das:SPR:u}
\end{align}
where $\Phi = \diag\{\phi_1,\dots,\phi_m\}$ and $\Lambda = \diag\{\lambda_1,\dots,\lambda_m\}$ are both positive definite gain matrices, i.e., $\phi_i,\lambda_i>0,\forall i$. Here, $\mathbb{I}\{\bar{e}\}$ denotes the elementwise time integral\footnote{$\mathbb{I}\{\bar{e}\} \coloneqq \displaystyle\int \bar{e}(\tau)\;d\tau$} of $\bar{e}$. The closed-loop transfer function, indexed\footnote{The subscript $u$ tags all quantities under the (unshaped) control input in \eqref{das:SPR:u} in the SPR analysis—e.g., $\mathbf{W}_u$, $\mathbf{C}_u$, $\Pi_u$, $r_u$, $\Omega_u$, etc. In the non–SPR analysis given in Section~\ref{sec:nonSPR}, the subscript $w$ designates shaped-objects, e.g., $\mathbf{W}_w$, $\mathbf{C}_w$, $\Pi_w$, $r_w$, $\Omega_w$. This convention applies to scalars, vectors, and matrices, and is used solely to separate the SPR and the non–SPR analyses and avoid symbol overloading.} by $u$, associated with \eqref{das:SPR:dynamics}–\eqref{das:SPR:u} can be written as
\begin{align}
    \begin{aligned}
        \mathbf{W}_u &= \Phi(s I_m + \Lambda)^2\,\Pi_u^{-1} \\
        &= \mathbf{C}_u\bigl[J(s^2 I_m) + B(s I_m) + \mathbf{C}_u\bigr]^{-1},
    \end{aligned}
    \label{das:SPR:Wu}
\end{align}
where
\begin{align*}
    \Pi_u = J(s^3 I_m) + (B+\Phi)(s^2 I_m) + 2\Phi\Lambda(s I_m) + \Phi\Lambda^2,
\end{align*}
and $\mathbf{C}_u$ denotes the Laplace transform of $c_u$ defined in \eqref{das:SPR:u}. Note that we design $\bar{u}$ so that $\mathbf W_u$ is SPR according to Definition~\ref{def:SPR}; \eqref{das:SPR:u} is one such choice.
\begin{definition}\label{def:SPR}
A real-rational transfer matrix $\mathbf W(s)$ is \emph{positive real (PR)} if:
\begin{itemize}
  \item $\mathbf W(s)$ is analytic in $\Re(s)>0$; and
  \item the Hermitian part is positive semidefinite,
        $\mathbf W(s)+\mathbf W^\ast (s)\succeq 0$, $\forall s$ with $\Re(s)>0$.
\end{itemize}
It is \emph{strictly positive real (SPR)} if there exists $\epsilon>0$ such that $\mathbf W(s-\epsilon)$ is PR.
\end{definition}

The transfer $\mathbf{W}_u=\diag\{W_{u,i}\}$ is diagonal, each $W_{u,i}$ is proper with relative degree~$1$, and is stable if
\begin{align}
    \Phi \succ \tfrac{1}{2}J\Lambda - B, \label{das:SPR:range:u}
\end{align}
equivalently $\,2(b_i+\phi_i) > J_i\lambda_i\,$ for all $i=1,\dots,m$. Condition \eqref{das:SPR:range:u} (see Remark~\ref{rem:RouthHurwitz}) is the Routh--Hurwitz inequality for the cubic denominator, hence $W_{u,i}$ is analytic in $\Re(s)>0$. For scalar $W_{u,i}$ of relative degree~$1$, Definition~\ref{def:SPR} reduces to: $W_{u,i}$ is analytic in $\Re(s)\!>\!0$ and $\Re\{W_{u,i}(j\omega)\}\!\ge\! 0$ for all $\omega$ (PR), with strict inequality for SPR. Under \eqref{das:SPR:range:u}, the even-polynomial representation of $\Re\{W_{u,i}(j\omega)\}$ is strictly positive for all $\omega$; therefore $\Re\{W_{u,i}(j\omega)\}>0$. Consequently, each $W_{u,i}$ is SPR, and thus $\mathbf W_u$ is SPR.

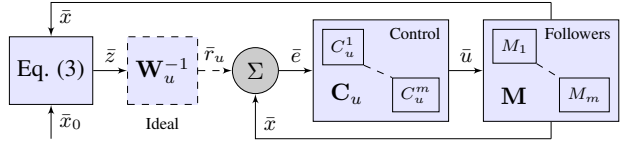
\begin{figure}[t!]
    \centering
    \scalebox{0.9}{\begin{tikzpicture}[auto,>=latex']
    \node [block] (Net) {Eq. \eqref{das:SPR:z}};
    % \node [anchor=north] at ($(Net.north)$) {\scriptsize Network};

    \node [block, dashed, right = .5cm of Net] (Ideal) {$\mathbf{W}_u^{-1}$};
    \node [below=2pt of Ideal, font=\scriptsize] {Ideal};
    \node [sum, right = .5cm of Ideal] (Sum) {$\Sigma$};
    \node [block, right = .5cm of Sum, minimum width=2cm, minimum height=1.5cm] 
    (Act) {};
    \node [block, right = .5cm of Act, minimum width=2cm, minimum height=1.5cm] 
    (Plant) {};
    \node [coordinate, below = .25cm of Plant] (C1) {};
    \node [coordinate, above = .25cm of Plant] (C2) {};
    \node [coordinate, below = .5cm of Net] (C3) {};
    
    % --- inner Plant ---
    \node[draw, rectangle, minimum width=5mm, minimum height=5mm,
          anchor=north west] (pA) at ($(Plant.north west)+(4pt,-4pt)$) {\scriptsize $M_1$};  
    \node[draw, rectangle, minimum width=5mm, minimum height=5mm,
          anchor=south east] (pB) at ($(Plant.south east)+(-4pt,4pt)$) {\scriptsize $M_m$};
    \draw[dashed] (pA.south east) -- (pB.north west);
    \node[anchor=south west] at ($(Plant.south west)+(4pt,4pt)$) {$\mathbf{M}$};
    \node[anchor=north east] at ($(Plant.north east)$) {\scriptsize Followers};

    % --- inner Act ---
    \node[draw, rectangle, minimum width=5mm, minimum height=5mm,
          anchor=north west] (pA) at ($(Act.north west)+(4pt,-4pt)$) {\scriptsize $C_u^{1}$};  
    \node[draw, rectangle, minimum width=5mm, minimum height=5mm,
          anchor=south east] (pB) at ($(Act.south east)+(-4pt,4pt)$) {\scriptsize $C_u^{m}$};
    \draw[dashed] (pA.south east) -- (pB.north west);
    \node[anchor=south west] at ($(Act.south west)+(4pt,4pt)$) {$\mathbf{C}_u$};
    \node[anchor=north east] at ($(Act.north east)$) {\scriptsize Control};
    
    \draw [draw, ->] (Net) -- node{\footnotesize $\bar{z}$} (Ideal);
    \draw [dashed, ->] (Ideal) -- node{\footnotesize $\bar{r}_u$} (Sum);
    \draw [draw, ->] (Sum) -- node{\footnotesize $\bar{e}$} (Act);
    \draw [draw, ->] (Act) -- node{\footnotesize $\bar{u}$} (Plant);
    \draw [draw, -] (Plant) -- (C1);
    \draw [draw, -] (Plant) -- (C2);
    \draw [draw, ->] (C1) -| node[above right]{\footnotesize $\bar{x}$} (Sum);
    \draw [draw, ->] (C2) -| node[below right]{\footnotesize $\bar{x}$} (Net);
    \draw [draw, ->] (C3) -- node[right]{\footnotesize $\bar{x}_0$} (Net);
    
    \end{tikzpicture}}
    \caption{Block diagram of \eqref{das:SPR:dynamics}–\eqref{das:SPR:Wu}. Solid lines depict the nominal loop. The dashed “Ideal” region indicates preprocessing $\mathbf{W}_u^{-1}$ acting on $\bar z$ (cf. \eqref{das:SPR:ref:u2}).}
    \label{Fig:ideal}
\end{figure}
\begin{remark}[Routh--Hurwitz Gain]\label{rem:RouthHurwitz}
    The inequality in \eqref{das:SPR:range:u} is equivalent to $\phi_i > \tfrac{1}{2}J_i\lambda_i - b_i, \forall i=1,\dots,m$. Thus, no ``upper bound'' on $\Phi$ is required: any larger $\phi_i$ preserves the condition and increases the margin. 
    A convenient \emph{uniform} sufficient design is $2\sigma_{\min}(\Phi) > \sigma_{\max}(J\Lambda - 2B)$, which guarantees \eqref{das:SPR:range:u} for all agents simultaneously. 
\end{remark}

Equations \eqref{das:SPR:dynamics}-\eqref{das:SPR:Wu} specify the setup; we now present the method. We further design the reference signals $\bar{r}_u$ for the followers as shown in Fig.~\ref{Fig:ideal}. If $\mathbf{W}_u$ were known, an ideal (exact-tracking) choice is
\begin{align}
    \bar{r}_u &\coloneqq \mathbf{W}_u^{-1}\bar{z} = \hat{J}\ddot{\Omega}_u + \hat{B}\dot{\Omega}_u + \bar{z}, \label{das:SPR:ref:u2}
\end{align}
where $\bar{z}$ is defined in \eqref{das:SPR:z} and $\Omega_u\coloneqq \mathbf{C}_u^{-1}\bar{z}$ is the internal signal generated by the inverse of the control law in \eqref{das:SPR:u}. Here $\hat{J}=\diag\{\hat{J}_1,\dots,\hat{J}_m\}$ and $\hat{B}=\diag\{\hat{b}_1,\dots,\hat{b}_m\}$ are parameter estimates used in place of the unknown $J,B$. The inverse of $\mathbf{W}_u$ in \eqref{das:SPR:ref:u2} is
\begin{align}
    \mathbf{W}_u^{-1} = I_m + \bigl[\hat{J}(s^2 I_m) + \hat{B}(s I_m)\bigr]\mathbf{C}_u^{-1}. \label{das:SPR:inv:Wu}
\end{align}
This construction “matches’’ the network signal $\bar{z}$ through the (estimated) plant dynamics so that, with exact parameters, the closed-loop transfer function is inverted and $\bar{e}=\bar{z}-\bar{x}\to \mathbf{0}_m$ (exact synchronization). 
However, the implementation of \eqref{das:SPR:ref:u2} requires $\dot{\Omega}_u$ and $\ddot{\Omega}_u$, which would entail exchanging neighbor derivatives, and violates position-only communication. 
The next lemma recasts the input into an implementable form; the subsequent remark specifies the substitutions.

\begin{lemma}\label{lem:Control}
    Consider $\bar{r}_u$ in \eqref{das:SPR:ref:u2}, the control law $\bar{c}_u(\bar{e})$ in \eqref{das:SPR:u}, a feedback signal $\bar{x}$, and a network $\mathcal{G}$ satisfying Remark~\ref{Rem:Threshold}. Then, $\bar{u} = \hat{J}\ddot{\bar{z}} + \hat{B}\dot{\bar{z}} + \mathbf{C}_u\bar{e}$. Under position-only exchange as in Remark~\ref{rem:Replace} yields the implementable form
    \begin{align}
        \bar{u} = \mathbf{C}_u\Bar{e} + \hat{J}\ddot{\Bar{x}}_0 + \hat{B}\dot{\Bar{x}}. \label{das:SPR:opt:u}
    \end{align}
\end{lemma}
\begin{proof}
    By definition, $\bar{u} \coloneqq \mathbf{C}_u(\bar{r}_u - \bar{x}) = \mathbf{C}_u \bar{r}_u - \mathbf{C}_u \bar{x}$. 
    Substituting \eqref{das:SPR:ref:u2}, using $\Omega_u=\mathbf{C}_u^{-1}\bar{z}$ and $\mathbf{C}_u \mathbf{C}_u^{-1} = I_m$, gives
    \begin{align*}
        \bar{u} &= \hat{J} \mathbf{C}_u \mathbf{C}_u^{-1} \ddot{\bar{z}} + \hat{B} \mathbf{C}_u \mathbf{C}_u^{-1} \dot{\bar{z}} + \mathbf{C}_u(\bar{z} - \bar{x}) \\
        &= \hat{J} \ddot{\bar{z}} + \hat{B} \dot{\bar{z}} + \mathbf{C}_u \bar{e}.
    \end{align*}
    Applying the replacements stated in Remark~\ref{rem:Replace} due to limited accessibility yields \eqref{das:SPR:opt:u}.
\end{proof}
\begin{remark}[Position-only Exchange]\label{rem:Replace}
    To preserve position-only communication, we substitute both neighbor derivatives $\dot{\bar{z}}\!=\!\mathbb{A}_m\dot{\bar{x}}+\mathbb{A}_0\dot{\bar{x}}_0$ and $\ddot{\bar{z}}\!=\!\mathbb{A}_m\ddot{\bar{x}}+\mathbb{A}_0\ddot{\bar{x}}_0$, with $\dot{\bar{x}}$ and $\ddot{\bar{x}}_0$, in turn, which are locally measurable for $\dot{\bar{x}}$ (moving the feedforward $\hat{B} \dot{\bar{z}}$ into the feedback $\hat{B} \dot{\bar{x}}$) and leader-disseminated for $\ddot{\bar{x}}_0$. 
\end{remark}

Building on Lemma~\ref{lem:Control}, we reparameterize the control law to incorporate the estimate $\hat{J}$, yielding standard gradient-type error dynamics while preserving the SPR structure (see Remark~\ref{rem:Reparameterization}). With this reparameterization, $\mathbf{C}_u$ is redefined as
\begin{align}
    \mathbf{C}_u
    &= \big[(\Phi + 2\hat{J}\Lambda)(s^2 I_m) + \big(2\Phi\Lambda + \hat{J}\Lambda^2\big)(s I_m) + \big.\Phi\Lambda^2\big](s I_m)^{-1}.
    \label{das:SPR:mod:Cu}
\end{align}
Correspondingly, the modified reference is
\begin{align}
    \bar{r}_u = \bar{z} + \hat{J}(s^2 I_m)\,\mathbf{C}_u^{-1}\bar{z}
    \eqqcolon \hat{J}\,\ddot{\Omega}_u + \bar{z}.
    \label{das:SPR:mod:ref:u}
\end{align}
Using \eqref{das:SPR:opt:u}, the follower input becomes
\begin{align}
    \bar{u}
    &= \mathbf{C}_u \bar{e} + \hat{J}\ddot{\bar{x}}_0 + \hat{B}\dot{\bar{x}} \nonumber\\
    &= (\Phi + 2\hat{J}\Lambda)\dot{\bar{e}}
       + \big(2\Phi\Lambda + \hat{J}\Lambda^2\big)\bar{e} + \Phi\Lambda^2\,\mathbb{I}\{\bar{e}\}
       + \hat{J}\ddot{\bar{x}}_0 + \hat{B}\dot{\bar{x}} \eqqcolon \Phi\vartheta_u + \hat{\Xi}_u\eta_u, \label{das:SPR:mod:u}
\end{align}
in which the two terms $\hat{\Xi}_u\in\mathbb{R}^{m\times 2m},\eta_u\in\mathbb{R}^{2m}$ are defined as $\hat{\Xi}_u = [\hat{J},\hat{B}]$ and $\eta_u = [\zeta_u^\top,\dot{\Bar{x}}^\top]^\top$ where $\zeta_u = \ddot{\Bar{x}}_0 + 2\Lambda\dot{\Bar{e}} + \Lambda^2\Bar{e}$ while the term $\vartheta_u\in\mathbb{R}^{m}$ is given by $\vartheta_u = \dot{\Bar{e}} + 2\Lambda\bar{e} + \Lambda^2\mathbb{I}\{\Bar{e}\}$. The idea of \eqref{das:SPR:mod:u} is to decouple the adaptive parameters $\hat{\Xi}_u$ from the time-varying regressor $\eta_u$ explicitly, while $\Phi$ shapes the error combination $\vartheta_u$ to be driven to zero. To proceed with the stability analysis, we define the parameter estimation error as $\tilde{\Xi}_u = \hat{\Xi}_u - \Xi_u$, and analyze the evolution of $\vartheta_u$ over the network characterized by $\mathbb{L}$ and $\mathbb{A}_0$. The corresponding error dynamics, expressed as the gradient of $\vartheta_u$, is derived as follows
\begin{subequations}
\begin{align}
    \dot{\vartheta}_u &= \ddot{\Bar{e}} + 2\Lambda\dot{\Bar{e}} + \Lambda^2\Bar{e} \label{das:SPR:Edyn:u1} \\
    &= -\mathbb{L}\ddot{\Bar{x}} + \mathbb{A}_0\ddot{\Bar{x}}_0 + 2\Lambda\dot{\Bar{e}} + \Lambda^2\Bar{e} \nonumber \\
    &= \ddot{\Bar{z}} - \ddot{\Bar{x}} + 2\Lambda\dot{\Bar{e}} + \Lambda^2\Bar{e} \label{das:SPR:Edyn:u2} \\
    &= \zeta_u + \bigl[B\dot{\Bar{x}} - \Phi\vartheta_u - \hat{J}\zeta_u - \hat{B}\dot{\Bar{x}}\bigr]J^{-1} \nonumber \\
    % &= \bigl[I_m - \hat{J}J^{-1}\bigr]\zeta_u - \bigl[\tilde{B}\dot{\Bar{x}} + \Phi\vartheta_u\bigr]J^{-1} \nonumber \\
    &= -J^{-1}\Phi\vartheta_u + J^{-1}\bigl[-\tilde{\Xi}_u\eta_u\bigr]. \label{das:SPR:Edyn:u3}
\end{align}
\end{subequations}
which shows the standard error dynamics for the networked adaptive system. In the following subsection, we analyze the stability of the follower dynamics based on \eqref{das:SPR:Edyn:u3}.
\begin{remark}[Why reparameterization?]\label{rem:Reparameterization}
     The reparameterization in \eqref{das:SPR:mod:Cu} brings $\hat{J}$ \emph{inside} the feedback operator via $\mathbf{C}_u$ so that the regressor $\zeta_u$ in \eqref{das:SPR:Edyn:u2} appears explicitly, resulting \eqref{das:SPR:mod:u}. Without \eqref{das:SPR:mod:Cu}, using $\mathbf{C}_u$ in \eqref{das:SPR:u}, one obtains $\bar{u} = \Phi \vartheta_u + \hat{J} \ddot{\bar{x}}_0 + \hat{B} \dot{\bar{x}}$ where $\zeta_u$ cannot be recovered and as a result, the terms $\ddot{\bar{x}}_0$ and $2\Lambda\dot{\bar{e}} + \Lambda\Lambda^\top\bar{e}$ cannot be grouped into a unified regressor, complicating the stability analysis. The structure introduced in \eqref{das:SPR:mod:Cu} is therefore a deliberate design choice, enabling a clean separation under position-only communication while preserving the \emph{SPR} design.
\end{remark}

\subsection{Stability Analysis}\label{subsec:Stability}

In the networked setting, our goal is to ensure parameter convergence and asymptotic synchronization, i.e., $x_i(t)\to x_0(t)$ for all $i$ (equivalently, $\bar e(t)\to \mathbf{0}_m$). Because the true parameters are unknown (no explicit form of $\tilde{\Xi}_u$), we work with the measurable error dynamics in \eqref{das:SPR:Edyn:u3}, capturing the influence of $\tilde{\Xi}_u$, and design distributed adaptive laws using only measurable signals.

\begin{theorem}\label{thm:SPR}
    Consider the follower dynamics \eqref{das:SPR:dynamics} controlled by \eqref{das:SPR:mod:u} over a communication network $\mathcal{G}$ satisfying Remark~\ref{Rem:Threshold}. Let the adaptive law \eqref{das:SPR:Lya:Xiu} be applied to the parameter estimates $\hat{\Xi}_u$ in \eqref{das:SPR:mod:u}. 
    Then the origin of the error system $(\vartheta_u,\tilde{\Xi}_u)=(\mathbf{0}_m,\mathbf{0})$ is globally asymptotically stable. 
    In particular, all closed-loop signals remain bounded and the synchronization error satisfies $\bar e(t)\to \mathbf{0}_m$ as $t\to\infty$. 
\end{theorem}
\begin{proof}
    Let $V_u(\vartheta_u,\tilde{\Xi}_u)>0$ be a Lyapunov function candidate inspired by the error dynamics in \eqref{das:SPR:Edyn:u3} of the form
    \begin{subequations}
    \begin{align}
        V_u = \vartheta_u^\top\mathbb{P}_u\vartheta_u + \tr\bigl[\tilde{\Xi}_u^\top J^{-1}\Gamma_\xi^{-1}\tilde{\Xi}_u\bigr]\label{das:SPR:Lya:Vu}
    \end{align}
    where $\Gamma_\xi \in \mathbb{R}^{m\times m} \succ 0$. From \eqref{das:SPR:Edyn:u3},
    \begin{equation*}
        \dot{\vartheta}_u = -J^{-1}\Phi\,\vartheta_u - J^{-1}\tilde{\Xi}_u\,\eta_u
        \eqqcolon \mathbf{A}_u\,\vartheta_u + \mathbf{B}_u\,\tilde{\Xi}_u\eta_u,
    \end{equation*}
    where $\mathbf{A}_u=-J^{-1}\Phi\prec0$ and $\mathbf{B}_u=-J^{-1}$. Since $\mathbf{A}_u$ is Hurwitz, there exist $\mathbb{Q}_u=\mathbb{Q}_u^\top\succ0$ and $\mathbb{P}_u=\mathbb{P}_u^\top\succ0$ solving the Lyapunov equation $\mathbf{A}_u^\top\mathbb{P}_u+\mathbb{P}_u\mathbf{A}_u=-\mathbb{Q}_u$. Taking the time derivative of \eqref{das:SPR:Lya:Vu} along the trajectories of \eqref{das:SPR:Edyn:u3}, we obtain
    \begin{align}
        \dot{V}_u &= \vartheta_u^\top\bigl[\mathbf{A}_u^\top\mathbb{P}_u + \mathbb{P}_u\mathbf{A}_u\bigr]\vartheta_u - 2\vartheta_u^\top\mathbb{P}_u J^{-1}\tilde{\Xi}_u\eta_u \nonumber \\ 
        &\quad + 2\tr\bigl[\dot{\tilde{\Xi}}_u^\top J^{-1}\Gamma_\xi^{-1}\tilde{\Xi}_u\bigr] \nonumber \\
        &= -\vartheta_u^\top\mathbb{Q}_u\vartheta_u - 2\vartheta_u^\top\mathbb{P}_u J^{-1}\tilde{J}\zeta_u - 2\vartheta_u^\top\mathbb{P}_u J^{-1}\tilde{B}\dot{\Bar{x}} \nonumber \\
        &\quad + 2\tr\bigl[\eta_u\vartheta_u^\top\mathbb{P}_u\Gamma_\xi J^{-1}\Gamma_\xi^{-1}\tilde{\Xi}_u\bigr] \nonumber \\
        &= -\vartheta_u^\top\mathbb{Q}_u\vartheta_u \leq 0 \label{das:SPR:Lya:DVu}
    \end{align}
    if the adaptive law is suitably chosen as $\dot{\tilde{\Xi}}_u = \Gamma_\xi\mathbb{P}_u\vartheta_u\eta_u^\top$. Alternatively, to account for heterogeneity in agent dynamics and potential asymmetries in network topologies, the term $\Gamma_\xi$ can be decoupled using independent gains as
    \begin{align} 
        \dot{\tilde{\Xi}}_u = \bigl[\Gamma_{\xi_1},\Gamma_{\xi_2}\bigr](I_2 \otimes \mathbb{P}_u \vartheta_u)\diag\{\zeta_u^\top,\dot{\bar{x}}^\top\} \label{das:SPR:Lya:Xiu} 
    \end{align} 
    where $\Gamma_{\xi_1},\Gamma_{\xi_2}\in\mathbb{R}^{m\times m}\succ 0$ serve as independent adaptation gains for $\hat{J}$ and $\hat{B}$, in turn, and $\diag\{\zeta_u^\top,\dot{\bar{x}}^\top\}\in\mathbb{R}^{2\times 2m}$. This formulation enables more flexible adaptation tailored to individual agent characteristics and improves robustness in non-uniform or sparse network configurations, where a shared adaptation rate may be suboptimal. Therefore $V_u(t)$ is nonincreasing and bounded below, i.e., $V_u(t) \leq V_u(0)$, and $\dot{V}_u \to 0$ as $t\to\infty$, implying $\vartheta_u\in\mathcal{L}_2\cap\mathcal{L}_\infty$ and $\tilde{\Xi}_u\in\mathcal{L}_\infty$. Since $\dot{\vartheta}_u$ is bounded under the closed-loop (from \eqref{das:SPR:Edyn:u3} with bounded signals), Barbalat’s lemma gives $\vartheta_u(t)\to 0$. By the definition of $\vartheta_u$ (a strictly proper, stable filter of $\bar e$), this implies $\bar e(t)\to \mathbf{0}_m$.
    \end{subequations}
\end{proof}

\begin{remark}\label{rem:LaSalle}
    By \eqref{das:SPR:dynamics}, \eqref{das:SPR:Edyn:u3}, \eqref{das:SPR:Lya:Vu}, and \eqref{das:SPR:Lya:DVu}, LaSalle’s invariance principle implies that trajectories converge to the largest invariant subset of
    \begin{equation*}
        \mathbb{S} = \bigl\{(\vartheta_u,\tilde{\Xi}_u): \dot V_u(\vartheta_u,\tilde{\Xi}_u)=0\bigr\}
        = \bigl\{(\vartheta_u,\tilde{\Xi}_u): \vartheta_u=\mathbf{0}_m, \dot{\tilde{\Xi}}_u=\mathbf{0}\bigr\}.
    \end{equation*}
    Thus $\vartheta_u\to\mathbf{0}_m$ (hence $\bar e\to\mathbf{0}_m$). If, in addition, the regressor $\eta_u$ is persistently exciting, then $\tilde{\Xi}_u\to \mathbf{0}$; otherwise $\tilde{\Xi}_u$ converges to a constant.
\end{remark}

\subsection{Exact Cancellation of Disturbance}\label{subsec:Disturbance}
Before proceeding, recall that for the disturbance-free case, the error dynamics in \eqref{das:SPR:Edyn:u3} can be rewritten as
\begin{align*}
    \begin{bmatrix}
        \dot{\vartheta}_u \\ \vect(\dot{\tilde{\Xi}}_u)
    \end{bmatrix} = \begin{bmatrix}
        \mathbf{A}_u & \mathbf{B}_u(I_m\otimes \eta_u^\top) \\
        \Gamma_{\xi}\otimes\mathbb{P}_u\eta_u & 0\\
    \end{bmatrix}\begin{bmatrix}
        \vartheta_u \\ \vect(\tilde{\Xi}_u)
    \end{bmatrix},
\end{align*}
which we compactly express as
\begin{align}
    \dot{\beta}_u = \mathbb{F}(\eta_u) \beta_u \label{das:SPR:stack}
\end{align}
where $\beta_u = [\vartheta_u^\top, \vect(\tilde{\Xi}_u)^\top]^\top$, and $\dot{\vartheta}_u = \mathbf{A}_u \vartheta_u + \mathbf{B}_u \tilde{\Xi}_u \eta_u$ with $\mathbf A_u=-J^{-1}\Phi$ and $\mathbf B_u=-J^{-1}$. We now examine the closed loop in the presence of exogenous disturbances. The follower dynamics in \eqref{das:SPR:dynamics} become
\begin{align}
    J\ddot{\bar{x}} + B\dot{\bar{x}} = \bar{u} + \bar{\delta}, \qquad \bar{\delta}\in\mathcal{L}_\infty\label{das:SPR:dynamics:dist}
\end{align}
where $\bar{\delta}=[\delta_1,\dots,\delta_m]^\top$ is a bounded disturbance (e.g., constant or time-varying; see Fig.~\ref{F3b}). With the control law \eqref{das:SPR:mod:u}, the error dynamics \eqref{das:SPR:Edyn:u3} acquire an additive input channel
\begin{equation}
    \dot{\bar{\vartheta}}_u = \mathbf A_u\,\bar{\vartheta}_u+\mathbf B_u\,\tilde{\Xi}_u\,\bar{\eta}_u+J^{-1}\bar{\delta}, \label{das:SPR:Edyn:dist}
\end{equation}
as shown in the Appendix. We distinguish two scenarios according to network connectivity and the leader signal.
\paragraph*{Case 1: disconnected network ($\bar x_0=\mathbf 0$, $\bar z=\mathbf 0, \forall t$).}
Here $\bar e=-\bar x$, so \eqref{das:SPR:mod:u} simplifies to $\Bar{u} = \Phi\bar{\vartheta}_u + \hat{\Xi}_u\bar{\eta}_u$ where $\bar{\vartheta}_u(\bar{x}) = -\dot{\Bar{x}} - 2\Lambda\bar{x} - \Lambda\Lambda^\top\mathbb{I}\{\Bar{x}\}$ and $\bar{\eta}_u = [\bar{\zeta}_u^\top,\dot{\bar{x}}^\top]^\top$ with $\bar{\zeta}_u = -2\Lambda\dot{\Bar{x}} - \Lambda\Lambda^\top\Bar{x}$. Considering \eqref{das:SPR:stack} and stacking $\bar\beta_u=\big[\bar{\vartheta}_u^\top,\vect(\tilde{\Xi}_u)^\top\big]^\top$ yield
\begin{equation}
    \dot{\bar\beta}_u=\mathbb F(\bar{\eta}_u)\,\bar\beta_u + \begin{bmatrix}J^{-1}\bar{\delta} \\ \mathbf 0_{2m^2}\end{bmatrix}.
    \label{eq:bar_beta_dist}
\end{equation}
In disconnected settings the regressor lacks richness; if $\bar x,\dot{\bar x}\notin\mathcal L_2$ (e.g., due to persistent inputs at the disturbance channel), then $\bar{\delta}\notin\mathcal L_2$. Without projection or $\sigma$–modification, $\vect(\tilde{\Xi}_u)$ can drift—this is the well-known lack-of-PE pathology in disconnected networks.
 
\paragraph*{Case 2: connected network ($\bar x_0\neq \mathbf 0,\ \bar z\neq \mathbf 0$).}
Using the Lyapunov function in \eqref{das:SPR:Lya:Vu} and the disturbed dynamics \eqref{das:SPR:Edyn:dist} which is an SPR-driven, strictly proper, stable filter of $\bar{\delta}$ in the $\vartheta_u$-channel, one obtains
\begin{equation}
    \dot V_u = -\vartheta_u^\top\mathbb Q_u\,\vartheta_u + 2\vartheta_u^\top\mathbb P_u J^{-1}\bar{\delta}
    \le -\tfrac{\lambda_{\min}(\mathbb Q_u)}{4}\,\|\vartheta_u\|^2
    + \kappa\,\|\bar{\delta}\|^2, \label{das:SPR:Lya:DVu:dist}
\end{equation}
for some $\kappa>0$ (Young’s inequality). Thus the $\vartheta_u$–subsystem is ISS with respect to $\bar{\delta}$. In particular,
\begin{itemize}
    \item If $\bar{\delta}$ is constant, then the integral term in \eqref{das:SPR:mod:u} provides an internal model for step signals, ensuring that $\mathbf{W}_u\bar{\delta}$ decays exponentially to zero. In particular, under the same adaptive law \eqref{das:SPR:Lya:Xiu} used in Theorem~\ref{thm:SPR}, the Lyapunov analysis carries through with an extra bounded input term, yielding $\vartheta_u\in\mathcal{L}_2\cap\mathcal{L}_\infty$ and $\bar{e}(t)\to \mathbf{0}_m$. Thus constant disturbances are exactly rejected (zero steady-state synchronization error).
    \item If $\bar{\delta}\in\mathcal L_2$, then $\vartheta_u\in\mathcal{L}_\infty$ and $\bar{e}\in\mathcal{L}_\infty$, with $\bar{e}(t)\to \mathbf{0}_m$.
    \item For bounded, time-varying $\bar{\delta}$, both $\vartheta_u$ and $\bar e$ are bounded. The steady-state synchronization error is attenuated according to the (stable) disturbance-to-error map induced by \eqref{das:SPR:Edyn:dist}; exact cancellation holds for the constant component, while higher-frequency content is filtered. When the leader trajectory $\bar{x}_0$ is sufficiently rich (persistently exciting in the standard sense\footnote{$\exists\,T,\alpha>0:\displaystyle \int_t^{t+T} \diag\{\bar{x}_0(\tau)\}^2 d\tau \succeq \alpha I_m,\ \forall t\ge t_0.$}), the regressor $\eta_u$ is PE, implying $\tilde{\Xi}_u(t)\to\mathbf 0$; without such richness, $\tilde{\Xi}_u$ remains bounded and converges to a constant.
\end{itemize}

Equations \eqref{das:SPR:Edyn:dist}–\eqref{das:SPR:Lya:DVu:dist} show that the SPR design ensures exact rejection of constant disturbances and input-to-state stability with respect to bounded time-varying disturbances, with parameter convergence guaranteed under persistent excitation of the leader signal.

\section{Passivity-Recovered (Non-SPR) Synchronization via Frequency Shaping}\label{sec:nonSPR}
Building on the \emph{passivity-certified} design in Section~\ref{sec:SPR}, we now address cases where the closed-loop transfer function is \emph{not} strictly positive real (non-SPR). The goal and the constraint remain the same: to synchronize all agents to a leader trajectory under dynamics in \eqref{das:SPR:dynamics}, with unknown $J$ and $B$. To illustrate a non-SPR case, we introduce a networked adaptive law based on phase-lead compensation denoted by $\bar{w}$, while preserving the same communication constraint (position exchange only; no neighbor velocities/accelerations).

To motivate our development, consider a controller of the form $\bar{w} = \bar{c}_w(\bar{e})$, with controller transfer function
\begin{align}
    \mathbf{C}_w = \Phi\mathbf{N}\mathbf{D}^{-1}, \quad p_i < q_i, \forall i \label{das:nonSPR:ctrl:w}
\end{align}
where $\mathbf{N} \coloneqq sI_m + \diag\{\bar{p}\}$ and $\mathbf{D} \coloneqq sI_m + \diag\{\bar{q}\}$ with $\bar{p} = [p_1,\dots,p_m]^\top$ and $\bar{q} = [q_1,\dots,q_m]^\top$. Although \eqref{das:nonSPR:ctrl:w} is attractive for transient shaping and phase margin improvement, the resulting closed-loop transfer function (denote the unshaped map by $\mathbf{W}_w^-$) is \emph{not} SPR even under nominal conditions, with $\bar{\delta}=0$ and fixed feedback $K=\diag\{\kappa_i\}$. The obstruction is structural: $\mathbf{C}_w$ is proper and the dynamics in \eqref{das:SPR:dynamics} has relative degree $n_d=2$, so $\mathbf{W}_w^-$ inherits relative degree $n_d>1$, violating the necessary SPR condition for the standard Lyapunov/adaptive proofs. Unlike the SPR case—where reparameterization (see \eqref{das:SPR:mod:Cu}-\eqref{das:SPR:Edyn:u3} and Remark~\ref{rem:Reparameterization}) yields the SPR transfer $\mathbf{W}_u$—a direct reparameterization with \eqref{das:nonSPR:ctrl:w} does \emph{not} fix this non-SPR issue; a more fundamental modification is required.

To resolve non-SPR while retaining distributed implementation, we develop two complementary designs:
Scenario 1 assumes local velocity availability and applies output shaping to construct an SPR transfer function; 
Scenario 2 enforces position-only exchange and shapes the regressor via stable filtering so that the induced map is SPR. 
Both preserve distributed implementability, require no prior knowledge of \(J,B\), and enable the same Lyapunov/adaptive analysis used in the SPR section after shaping. We analyze both scenarios in the disturbance-free setting; once SPR is established, disturbance handling follows directly from Subsection~\ref{subsec:Disturbance}.

\textbf{Scenario 1 (output shaping; $\dot{\bar{x}}$ available locally).} Define the shaped network outputs
\begin{align}
    \bar{y} = (sI_m + \Theta)\bar{x}, \qquad \Theta = \diag\{\theta_i\} \succ 0 \label{das:nonSPR:y}
\end{align}
where $\Theta$ defines tunable parameters and $\bar{y} = [y_1,\dots,y_m]^\top$. Given \eqref{das:SPR:dynamics}, \eqref{das:nonSPR:ctrl:w}, and a feedback gain $K^\ast$, the closed-loop transfer function is written as
\begin{align}
    \mathbf{W}_w = \mathbf{C}_w(sI_m + \Theta)\Pi_w^{-1}, \label{das:nonSPR:Ww}
\end{align}
where $\Pi_w = J(s^2I_m) + B(sI_m) + K^\ast \mathbf{C}_w(sI_m + \Theta)$ and more precisely, $K^\ast = \diag\{k_i^\ast\}$ is a gain for which $\mathbf{W}_w$ has desired phase margins. The factor $(sI_m+\bar \theta)$ in \eqref{das:nonSPR:y} lowers the relative degree of the transfer function so that $\mathbf W_w$ becomes SPR and admits a passivity-based interpretation, in which its internal model can be stabilized via Lyapunov arguments.

Furthermore, compared to \eqref{das:SPR:z} and \eqref{das:SPR:error}, we use the shaped consensus signals $\bar{e}_y = \bar{z}_y - \bar{y}$ due to the shaped variables in \eqref{das:nonSPR:y}. The meaning of $\bar{z}_y = \mathbb{A}_m\Bar{y} + \mathbb{A}_0\Bar{y}_0$ is similar to $\bar{z}$ and $\bar{y}_0 = \mathbf{1}_m\otimes y_0$, with $y_0=(s+\theta_0)x_0$ is the shaped leader signal for this setting as opposed to $\bar{x}_0$. The (ideal) reference signal assuming the known $\mathbf{W}_w$ is denoted as
\begin{align}
    \Bar{r}_w &= \mathbf{W}_w^{-1}\bar{z}_y \nonumber \\
    &= \hat{K}\Bar{z}_y + \bigl[\hat{J}(s^2I_m) + \hat{B}(sI_m)\bigr]\bigl[\mathbf{C}_w(sI_m + \Theta)\bigr]^{-1} \Bar{z}_y \nonumber \\
    &= \hat{J}\ddot{\Omega}_w + \hat{B}\dot{\Omega}_w + \hat{K}\Bar{z}_y, \label{das:nonSPR:ref:w}
\end{align}
where $\Omega_w = \bigl[\mathbf{C}_w(sI_m + \Theta)\bigr]^{-1} \Bar{z}_y \approx \mathbf{C}_w^{-1} \Bar{z}$ and $\hat{J}$, $\hat{B}$, $\hat{K}\coloneqq K^\ast + \tilde{K}$ are the estimates. Introduce
\begin{align*}
    \bar{e}_w \coloneqq \bar{r}_w - \hat{K}\bar{y} = \hat{J}\ddot{\Omega}_w + \hat{B}\dot{\Omega}_w + \hat{K}\Bar{e}_y,
\end{align*}
and define the parameter estimation error $\tilde{\Xi}_w = \hat{\Xi}_w - \Xi_w$, where $\hat{\Xi}_w = [\hat{K}, \hat{J}, \hat{B}]$ and $\Xi_w$ is the true parameter. Then
\begin{align} 
    \bar{e}_w = \tilde{\Xi}_w \eta_w + K^\ast \bar{e}_y + J\ddot{\Omega}_w + B\dot{\Omega}_w, \label{das:nonSPR:opt:w1} 
\end{align} 
where $\eta_w = [\bar{e}_y^\top, \ddot{\Omega}_{w}^\top, \dot{\Omega}_{w}^\top]^\top$. Consistent with position-exchange constraints (Remark~\ref{rem:Replace}), inaccessible derivatives of $\bar z_y$ are replaced by leader-based signals filtered through $\mathbf C_w^{-1}$
\begin{align} 
    \bar{e}_w &= \tilde{\Xi}_w \eta_w + K^\ast \bar{e}_y + J\mathbf{C}_w^{-1} \ddot{\bar{x}}_0 + B \mathbf{C}_w^{-1} \dot{\bar{x}}_0 \nonumber \\
    &= \tilde{\Xi}_w \eta_w + \underbrace{K^\ast \bar{e}_y + J \ddot{\Omega}_{x_0} + B \dot{\Omega}_{x_0}}_{\Xi_w \eta_w}, \label{das:nonSPR:opt:w2}
    % &= \tilde{\Xi}_w \eta_w + \Xi_w \eta_w \label{das:nonSPR:opt:w2} 
\end{align} 
where $\ddot{\Omega}_{x_0} = \mathbf{C}_w^{-1} \ddot{\bar{x}}_0$ and $\dot{\Omega}_{x_0} = \mathbf{C}_w^{-1} \dot{\bar{x}}_0$, with the resulting time varying regressor $\eta_w = [\bar{e}_y^\top, \ddot{\Omega}_{x_0}^\top, \dot{\Omega}_{x_0}^\top]^\top$. To certify SPR, we invoke the Meyer–Kalman–Yakubovich (MKY) lemma, which links the condition in Definition~\ref{def:SPR} to the existence of a Lyapunov function; proof in \cite{R18}.
\begin{lemma}[MKY]\label{lem:MKY}
    Consider a transfer function $\mathbf{W}(s)$, a matrix $\gamma I_m, \gamma>0$ and a symmetric matrix $\varphi\succ 0$. If 
    \begin{align*}
        \Re\bigl[\mathbf{W}(i\omega)\bigr] &\triangleq \Re\bigl[\tfrac{\gamma}{2}I_m + \mathbf{W}\bigr]  > 0, \forall \omega\in\mathbb{R}
    \end{align*}
    where $\mathbf{W} = \mathbf{c}^\top(i\omega I_{mn} - \mathbf{A})^{-1}\mathbf{B}$, then $\exists \epsilon > 0, \alpha \coloneqq \diag\{\alpha_1,\dots,\alpha_m\}, \alpha_i\in\mathbb {R}^n, \mathbb{P}\succ 0$, such that $\mathbf{A}^\top\mathbb{P} + \mathbb{P}\mathbf{A} = -\alpha\alpha^\top - \epsilon\varphi$ and $\mathbb{P}\mathbf{B} - \mathbf{c} = \sqrt{\gamma}\alpha$.
\end{lemma}
From \eqref{das:nonSPR:opt:w2}, the optimal term $\Xi_w\eta_w$ corresponds to the ideal feedforward control that drives $\bar{e}_y \coloneqq \mathbf{W}_w\Xi_w\eta_w$ to zero, ensuring consensus. In contrast, the non-optimal term $\tilde{\Xi}_w \eta_w$ must be shown to decay over time. The error model can be derived as $\bar{e}_y \coloneqq \mathbf{W}_w\tilde{\Xi}_w\eta_w$, which is equivalent to the state-space $\dot{\bar{e}}_x = \mathbf{A}_w\bar{e}_x + \mathbf{B}_w\tilde{\Xi}_w\eta_w$ and $\bar{e}_y = \mathbf{c}_w^\top\bar{e}_x$ in which $\bar{e}_x$ represents the internal state. Therefore, based on Lemma~\ref{lem:MKY}, if $\mathbf{W}_w$ in \eqref{das:nonSPR:Ww} is SPR, then $\exists\mathbb{Q}_w=\mathbb{Q}_w^\top\succ 0$ such that $\mathbb{P}_w=\mathbb{P}_w^\top\succ 0$ solves the equation of $\mathbf{A}_w^\top\mathbb{P}_w + \mathbb{P}_w\mathbf{A}_w = -\mathbb{Q}_w$ with $\mathbb{P}_w\mathbf{B}_w = \mathbf{c}_w$. With the Lyapunov function
\begin{subequations}
\begin{align}
    V_w(\bar{e}_x,\tilde{\Xi}_w) = \bar{e}_x^\top\mathbb{P}_w\bar{e}_x + \tr\bigl[\tilde{\Xi}_w^\top\Gamma_\xi^{-1}\tilde{\Xi}_w\bigr], \label{das:nonSPR:Lya:Vw}
\end{align}
its derivative along $\bar{e}_y = \mathbf{W}_w\tilde{\Xi}_w\eta_w$ becomes
\begin{align}
    \dot{V}_w &= \bar{e}_x^\top\bigl[\mathbf{A}_w^\top\mathbb{P}_w + \mathbb{P}_w\mathbf{A}_w\bigr]\bar{e}_x + 2\bar{e}_x^\top\mathbb{P}_w\mathbf{B}_w\tilde{\Xi}_w\eta_w + 2\tr\bigl[\dot{\tilde{\Xi}}_w^\top\Gamma_\xi^{-1}\tilde{\Xi}_w\bigr] \nonumber \\
    &= -\bar{e}_x^\top\mathbb{Q}_w\bar{e}_x + 2\bar{e}_x^\top\mathbb{P}_w\mathbf{B}_w\tilde{\Xi}_w\eta_w + 2\tr\bigl[\eta_w\bar{e}_y^\top\Gamma_\xi\Gamma_\xi^{-1}\tilde{\Xi}_w\bigr] \nonumber \\
    &= -\bar{e}_x^\top\mathbb{Q}_w\bar{e}_x \leq 0, \label{das:nonSPR:Lya:DVw}
\end{align}
if the adaptive law $\Tilde{\Xi}_w$ is suitably chosen as 
\begin{align}
    \dot{\tilde{\Xi}}_w = -[\Gamma_{\xi_1},\Gamma_{\xi_2},\Gamma_{\xi_3}](I_3\otimes\bar{e}_y)\diag\{\bar{e}_y^\top,\ddot{\Omega}_{x_0}^\top,\dot{\Omega}_{x_0}^\top\} \label{das:nonSPR:Lya:Xiw}
\end{align}
with diagonal gains $\Gamma_{\xi_k}\succ 0$, $k = 1,2,3$.
Hence \(\bar e_x,\bar e_y,\tilde{\Xi}_w\in\mathcal L_\infty\), \(\bar e_x\in\mathcal L_2\), and \(\bar e_x(t)\to0\). Since the transfer $\mathbf{W}_w$ is strictly proper and stable, \(\bar e_y(t)\to \mathbf 0_m\), i.e., synchronization holds in the disturbance-free setting after shaping.
\end{subequations}

\textbf{Scenario 2 (regressor shaping; only $\bar{x}$ available).}
When only $\bar x$ is measurable, the output shaping $(sI_m+\Theta)$ of \eqref{das:nonSPR:y} cannot be implemented in the shared signals. The closed-loop transfer function reduces to $\mathbf{W}_w^- = \mathbf{C}_w (\Pi_w^-)^{-1}$ where $\Pi_w^- = J(s^2I_m) + B(sI_m) + K^\ast \mathbf{C}_w$. Consequently, the (ideal) references \eqref{das:nonSPR:ref:w} and errors \eqref{das:nonSPR:opt:w2} are formed with $\bar z$ and $\bar e$ (rather than $\bar z_y,\bar e_y$). In particular, the time-varying regressor becomes $\eta_w = [\bar{e}^\top, \ddot{\Omega}_{x_0}^\top, \dot{\Omega}_{x_0}^\top]^\top$ and, nonetheless, the error model $\bar{e} = \mathbf{W}_w^-\tilde{\Xi}_w\eta_w$ is \emph{not} SPR (relative degree $n_d >1$). To recover passivity, we shape the regressor by a stable pre-filter and use its inverse inside the adaptation:
\begin{align}
    \bar{e} &= \mathbf{W}_w^-(sI_m + \Theta)\tilde{\Xi}_w\bar{\eta}_w \label{das:nonSPR:opt:w3}
\end{align}
where $\bar{\eta}_w = \eta_w(sI_m + \Theta)^{-1}$. The composite transfer $\mathbf W_w \coloneqq \mathbf W_w^-(sI_m+\Theta)$ has $n_d = 1$ and is rendered SPR. Hence we can reuse the Lyapunov function in \eqref{das:nonSPR:Lya:Vw} (now with the transfer $\mathbf W_w$) and obtain the same dissipation inequality as in \eqref{das:nonSPR:Lya:DVw}.
Choosing a uniform gradient law $\dot{\tilde{\Xi}}_w = -\Gamma_{\xi}\bar e\bar\eta_w^\top$, with $\Gamma_{\xi} \succ 0$ yields $\dot V_w\le 0$, implying $\bar e_x,\bar e\in\mathcal L_\infty$, $\bar e_x\in\mathcal L_2$, and $\bar e_x(t)\to0$. Since $\mathbf W_w$ is strictly proper and stable, $\bar e(t)\to\mathbf 0_m$, establishing synchronization.

\section{Numerical Simulations}\label{sec:NumRes}

This section validates the distributed adaptive control framework across (i) unstable heterogeneous dynamics, (ii) multiple graph topologies, and (iii) diverse disturbance profiles. The dynamics of the followers are modeled as in \eqref{das:SPR:dynamics:dist} with unknown $J, B$ and the leader trajectory $x_0(t)=\sin(t)+0.75\cos(2t)$. The parameters of the unstable heterogeneous followers used in the simulations are as follows:
\begin{align}
    J_i = 0.5 + 0.1i,\quad b_i = -1.3 - 0.1i, \label{das:NumRes:Par}
\end{align}
with control gains $\phi_i = 1.5 + 0.5i$ and $\lambda_i=1$ for $i=1,\dots,8$, and adaptation gains $\Gamma_{\xi_k}=5I_m$, $k=1,2,3$. Three different disturbances are considered to evaluate the robustness of the control designs: $\bar{\delta}_1 = 0.25$, $\bar{\delta}_2 = 0.45 \sin(2t)$, and $\bar{\delta}_3 = 0.7 \sin(2t) + 0.3 \cos(3t)$ as shown in Fig.~\ref{F3b}. Four topologies (see Fig.~\ref{Fig:topology}), all satisfying Remark~\ref{Rem:Threshold}, are used with the following edge weights:
\begin{enumerate}
% \begin{enumerate}[leftmargin=*]
    \item \textbf{Star}: $w_{i0}=1,\ \forall i$.
    \item \textbf{Cyclic}: $w_{i0}=0.5$ and $w_{ij}=0.25$ for $j\in\mathcal N_i\!\setminus\!\{0\}$.
    \item \textbf{Series (path)}: edges $(i,i-1)$ with $w_{ij}=1$.
    \item \textbf{Arbitrary}: as in Table~\ref{Table:WeightArbitrary}.
\end{enumerate}
\begin{table}[h!]
    \centering
    \caption{Weights used in arbitrary network}
        \begin{adjustbox}{width=.7\columnwidth}
        \begin{tabular}{c|l|c|l}
            \toprule
            Agent & Weights & Agent & Weights \\
            \midrule
            1 & $w_{10}=1$ & 5 & $w_{52}=.25$, $w_{57}=.75$\\
            2 & $w_{20}=.75$, $w_{24}=.25$ & 6 & $w_{63}=.75$, $w_{68}=.25$\\
            3 & $w_{30}=1$ & 7 & $w_{75}=.5$, $w_{76}=.45$, $w_{78}=.05$\\
            4 & $w_{41}=.5$, $w_{42}=.5$ & 8 & $w_{86}=.5$, $w_{87}=.5$\\
            \bottomrule
        \end{tabular}
        \end{adjustbox}
    \label{Table:WeightArbitrary}
\end{table}

\begin{figure}[t!]
    \centering
    \scalebox{0.75}{{\begin{tikzpicture}
            \centering
            % \draw[black] (-1.5,-1.5) rectangle (1.5,1.5);
            \Text[x=0,y=1.5,fontsize=\small]{Star};
            \Vertex[x=0,y=0,label=$0$,color=red,fontcolor=white,size=.5]{L}
            \Vertex[x=1,y=0,label=$1$,color=blue,fontcolor=white,size=.5]{1}
            \Vertex[x=1,y=1,label=$2$,color=blue,fontcolor=white,size=.5]{2}
            \Vertex[x=0,y=1,label=$3$,color=blue,fontcolor=white,size=.5]{3}
            \Vertex[x=-1,y=1,label=$4$,color=blue,fontcolor=white,size=.5]{4}
            \Vertex[x=-1,y=0,label=$5$,color=blue,fontcolor=white,size=.5]{5}
            \Vertex[x=-1,y=-1,label=$6$,color=blue,fontcolor=white,size=.5]{6}
            \Vertex[x=0,y=-1,label=$7$,color=blue,fontcolor=white,size=.5]{7}
            \Vertex[x=1,y=-1,label=$8$,color=blue,fontcolor=white,size=.5]{8}
            \Edge[Direct,color=red](L)(1)
            \Edge[Direct,color=red](L)(2)
            \Edge[Direct,color=red](L)(3)
            \Edge[Direct,color=red](L)(4)
            \Edge[Direct,color=red](L)(5)
            \Edge[Direct,color=red](L)(6)
            \Edge[Direct,color=red](L)(7)
            \Edge[Direct,color=red](L)(8)
        \end{tikzpicture}}}
    \,
    \scalebox{0.75}{{\begin{tikzpicture}
            \centering
            % \draw[black] (-1.5,-1.5) rectangle (1.5,1.5);
            \Text[x=0,y=1.5,fontsize=\small]{Cyclic};
            \Vertex[x=0,y=0,label=$0$,color=red,fontcolor=white,size=.5]{L}
            \Vertex[x=1,y=0,label=$1$,color=blue,fontcolor=white,size=.5]{1}
            \Vertex[x=1,y=1,label=$2$,color=blue,fontcolor=white,size=.5]{2}
            \Vertex[x=0,y=1,label=$3$,color=blue,fontcolor=white,size=.5]{3}
            \Vertex[x=-1,y=1,label=$4$,color=blue,fontcolor=white,size=.5]{4}
            \Vertex[x=-1,y=0,label=$5$,color=blue,fontcolor=white,size=.5]{5}
            \Vertex[x=-1,y=-1,label=$6$,color=blue,fontcolor=white,size=.5]{6}
            \Vertex[x=0,y=-1,label=$7$,color=blue,fontcolor=white,size=.5]{7}
            \Vertex[x=1,y=-1,label=$8$,color=blue,fontcolor=white,size=.5]{8}
            \Edge[Direct,style={dashed},color=red](L)(1)
            \Edge[Direct,style={dashed},color=red](L)(2)
            \Edge[Direct,style={dashed},color=red](L)(3)
            \Edge[Direct,style={dashed},color=red](L)(4)
            \Edge[Direct,style={dashed},color=red](L)(5)
            \Edge[Direct,style={dashed},color=red](L)(6)
            \Edge[Direct,style={dashed},color=red](L)(7)
            \Edge[Direct,style={dashed},color=red](L)(8)
            \Edge[](1)(2)
            \Edge[](2)(3)
            \Edge[](3)(4)
            \Edge[](4)(5)
            \Edge[](5)(6)
            \Edge[](6)(7)
            \Edge[](7)(8)
            \Edge[](8)(1)
        \end{tikzpicture}}}
    \,
    \scalebox{0.75}{{\begin{tikzpicture}
            \centering
            % \draw[black] (-1.5,-1.5) rectangle (1.5,1.5);
            \Text[x=0,y=1.5,fontsize=\small]{Series};
            \Vertex[x=0,y=0,label=$0$,color=red,fontcolor=white,size=.5]{L}
            \Vertex[x=1,y=0,label=$1$,color=blue,fontcolor=white,size=.5]{1}
            \Vertex[x=1,y=1,label=$2$,color=blue,fontcolor=white,size=.5]{2}
            \Vertex[x=0,y=1,label=$3$,color=blue,fontcolor=white,size=.5]{3}
            \Vertex[x=-1,y=1,label=$4$,color=blue,fontcolor=white,size=.5]{4}
            \Vertex[x=-1,y=0,label=$5$,color=blue,fontcolor=white,size=.5]{5}
            \Vertex[x=-1,y=-1,label=$6$,color=blue,fontcolor=white,size=.5]{6}
            \Vertex[x=0,y=-1,label=$7$,color=blue,fontcolor=white,size=.5]{7}
            \Vertex[x=1,y=-1,label=$8$,color=blue,fontcolor=white,size=.5]{8}
            \Edge[Direct,color=red](L)(1)
            \Edge[Direct](1)(2)
            \Edge[Direct](2)(3)
            \Edge[Direct](3)(4)
            \Edge[Direct](4)(5)
            \Edge[Direct](5)(6)
            \Edge[Direct](6)(7)
            \Edge[Direct](7)(8)
        \end{tikzpicture}}}
    \,
    \scalebox{0.75}{{\begin{tikzpicture}
            \centering
            % \draw[black] (-1.5,-1.5) rectangle (1.5,1.5);
            \Text[x=0.5,y=1.5,fontsize=\small]{Arbitrary};
            \Vertex[x=0,y=1,label=$0$,color=red,fontcolor=white,size=.5]{L}
            \Vertex[x=1,y=0,label=$3$,color=blue,fontcolor=white,size=.5]{3}
            \Vertex[x=2,y=-1,label=$7$,color=blue,fontcolor=white,size=.5]{7}
            \Vertex[x=2,y=0,label=$6$,color=blue,fontcolor=white,size=.5]{6}
            \Vertex[x=2,y=1,label=$8$,color=blue,fontcolor=white,size=.5]{8}
            \Vertex[x=-1,y=0,label=$1$,color=blue,fontcolor=white,size=.5]{1}
            \Vertex[x=-1,y=-1,label=$4$,color=blue,fontcolor=white,size=.5]{4}
            \Vertex[x=0,y=-1,label=$2$,color=blue,fontcolor=white,size=.5]{2}
            \Vertex[x=1,y=-1,label=$5$,color=blue,fontcolor=white,size=.5]{5}
            \Edge[Direct,color=red](L)(1)
            \Edge[Direct,style={dashed},color=red](L)(2)
            \Edge[Direct,color=red](L)(3)
            \Edge[Direct](1)(4)
            \Edge[Direct](2)(5)
            \Edge[Direct](3)(6)
            \Edge[Direct](6)(7)
            \Edge[Direct](6)(8)
            \Edge[color=green,label=I](2)(4)
            \Edge[color=green,label=II](5)(7)
            \Edge[bend = -45,color=green,label=III](7)(8)
        \end{tikzpicture}}}
    \caption{Four network topologies are considered adhering to Remark~\ref{Rem:Threshold}. The blue circles with a number in each topology indicate the followers while the arrows show the communication network. In general, the red arrows express to which follower the leader trajectory is delivered $(\mathcal{G}_0)$ while the black arrows are the links among the followers $(\mathcal{G}_m)$. More specifically, the red dashed arrows in cyclic and arbitrary network mean that we perform in-depth analysis in various scenarios while the green arrows in arbitrary network shows the exchange information from various levels of followers from the leader $x_0$.}
    \label{Fig:topology}
\end{figure}
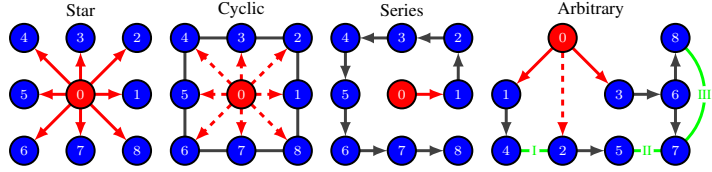
\begin{table}[t!]
  \centering
  \caption{Scaling of simulation runtime with the number of agents \(m\).
  Entries are median wall-clock time per run (seconds) over 20 runs.
  Hardware (MacBook Pro, Apple M2, 8\,GB) and solver settings (MATLAB) fixed across runs.}
  \label{Table:Scaling}
  \begin{adjustbox}{width=0.5\columnwidth}
  \begin{tabular}{r|ccc|ccc}
    \toprule
    & \multicolumn{3}{c}{SPR} & \multicolumn{3}{c}{non-SPR (Scenario 1)} \\
    \cmidrule(lr){2-4}\cmidrule(lr){5-7}
    \(\,m\,\) & Star & Cyclic & Path & Star & Cyclic & Path \\
    \midrule
     50  & 0.0201 & 0.0235 & 0.0793 & 0.0193 & 0.0272 & 0.0781 \\
    100  & 0.0341 & 0.0320 & 0.1692 & 0.0313 & 0.0490 & 0.1592 \\
    150  & 0.1056 & 0.0909 & 0.7605 & 0.1036 & 0.1005 & 0.8655 \\
    200  & 0.1184 & 0.1273 & 1.1563 & 0.1271 & 0.1537 & 1.2065 \\
    250  & 0.1325 & 0.1507 & 1.7258 & 0.1511 & 0.1773 & 1.8385 \\
    \bottomrule
  \end{tabular}
  \end{adjustbox}
\end{table}

We report results in three parts: (i) SPR (passivity-certified) synchronization—disturbance rejection and network-level behavior across topologies; (ii) non-SPR (passivity-recovered via frequency shaping) synchronization; and (iii) scalability, via runtime benchmarking for SPR and non-SPR settings as \(m\) increases.

\begin{itemize}[leftmargin=*]
    \item \textbf{Disturbance response (single agent).} Figure~\ref{F3a} reports the response for the most challenging follower ($i=8$) under $\bar{\delta}_k$, $k=1,2,3$. In all cases the follower synchronizes to $x_0$, while $\hat{J}$ and $\hat{B}$ converge, confirming robustness and disturbance rejection of the \emph{passivity-certified} design.
    
    \item \textbf{Network performance (SPR).} Figure~\ref{F4} shows network-level synchronization across star, cyclic, path, and arbitrary graphs under the worst disturbance $\bar{\delta}_3$. Despite large variations in connectivity and information flow, all followers track $x_0$. Hence, the SPR transfer $\mathbf W_u$ ensures uniform convergence margins across heterogeneous topologies.
    
    \item \textbf{Influence of cyclic leader weights (SPR).} Figure~\ref{F5a} varies the direct leader weight $w_{i0}$ in the cyclic graph. For $w_{i0}\!\ge\!0.15$, synchronization is achieved; at $w_{i0}=0.05$ the network drifts toward follower-only consensus, revealing insufficient external excitation. To probe partial access, Fig.~\ref{F5b} fixes a minimal leader weight $w_{i0}=0.15$ on only $i^c$ agents. Even with $i^c\le 5$, the network synchronizes efficiently—performance comparable to larger leader weights $(w_{i0}\in\{0.95,0.75,0.5\})$. Thus, modest, well-placed leader links suffice to preserve synchronization.
    
    \item \textbf{Arbitrary topology robustness (SPR).} Figure~\ref{F6a} varies a leader edge ($w_{20}\in\{0.05,0.15,0.25,0.5\}$) in the arbitrary graph. Unlike the cyclic case, all configurations synchronize to $x_0$, even with $w_{20}=0.05$, highlighting the resilience of irregular connectivity via alternative paths. Figure~\ref{F6b} then removes selected links (sets I–III in Fig.~\ref{Fig:topology}), violating $(\mathbb L-\mathbb A_0)\mathbf 1_m=0$; synchronization persists and the error ($\mathcal L_2$-norm) remains bounded, showing that the SPR design tolerates structural perturbations without retuning.
    
    \item \textbf{Passivity-recovered (non-SPR) synchronization.} Figures~\ref{F7a}–\ref{F7b} evaluate the \emph{passivity-recovered} controllers. Both implementations—shaped output combinations and position-only feedback—match the SPR performance across all topologies under $\bar{\delta}_3$. This confirms that frequency shaping recovers effective passivity of (unshaped) transfer function, making the overall loop satisfy the same Lyapunov-based guarantees and delivering synchronization comparable to the SPR case.

    % \item \textbf{Complexity Analysis.} Per update, each agent performs \(O(1)\) local compute and exchanges positions with its in-neighbors \(\bigl(O(|\mathcal N_i|)\bigr)\) communication. Aggregated over the network, the compute cost is \(O(m)\) and the communication cost is \(O(|\mathcal E|)\): linear in \(m\) for sparse graphs (bounded average degree) and \(O(m^2)\) for dense graphs. Empirical runtimes for both SPR and non-SPR settings as \(m\) increases corroborate these trends; see Table~\ref{Table:Scaling}.

    \item \textbf{Complexity Analysis.} Per update, each agent does \(O(1)\) local compute and \(O(\mathcal N_i)\) communication. Network-wide, compute is \(O(m)\) and communication is \(O(\mathcal E)\)—linear in \(m\) for sparse graphs, \(O(m^2)\) for dense ones. Empirical runtimes for SPR and non-SPR setting corroborate this; see Table~\ref{Table:Scaling}.
\end{itemize}

Across unstable heterogeneous dynamics, disturbances $\bar{\delta}_1$–$\bar{\delta}_3$, and diverse graphs, the \emph{passivity-certified} (SPR) design yields exact step disturbance rejection and bounded responses to time-varying inputs, with parameter convergence under persistent excitation. The \emph{passivity-recovered} (non-SPR) design achieves comparable tracking after frequency shaping. Taken together, the simulations support the central claim: a passivity-agnostic, distributed adaptive framework can deliver reliable synchronization in large-scale networks, whether the closed loop is inherently SPR or rendered SPR by design.

\begin{figure}[t!]
    \centering
    \subfloat[\label{F3a}\scriptsize Performance and Convergence]{\includegraphics[width=0.4\linewidth]{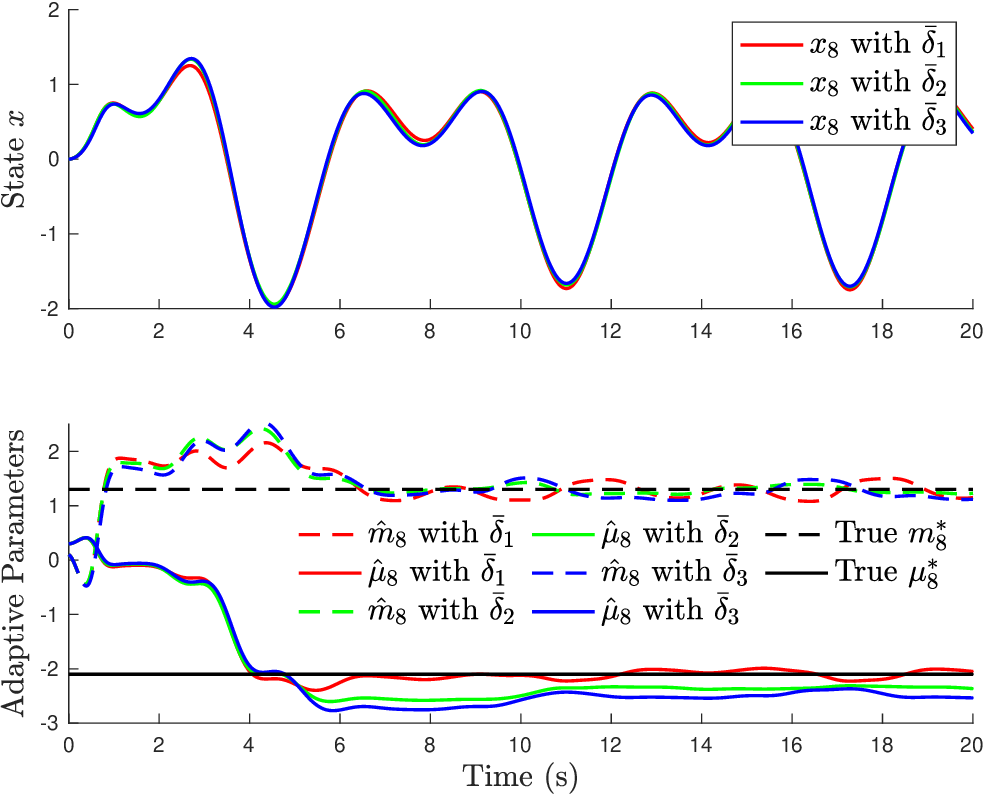}}\qquad
    \subfloat[\label{F3b}\scriptsize Disturbance Profiles]{\includegraphics[width=0.4\linewidth]{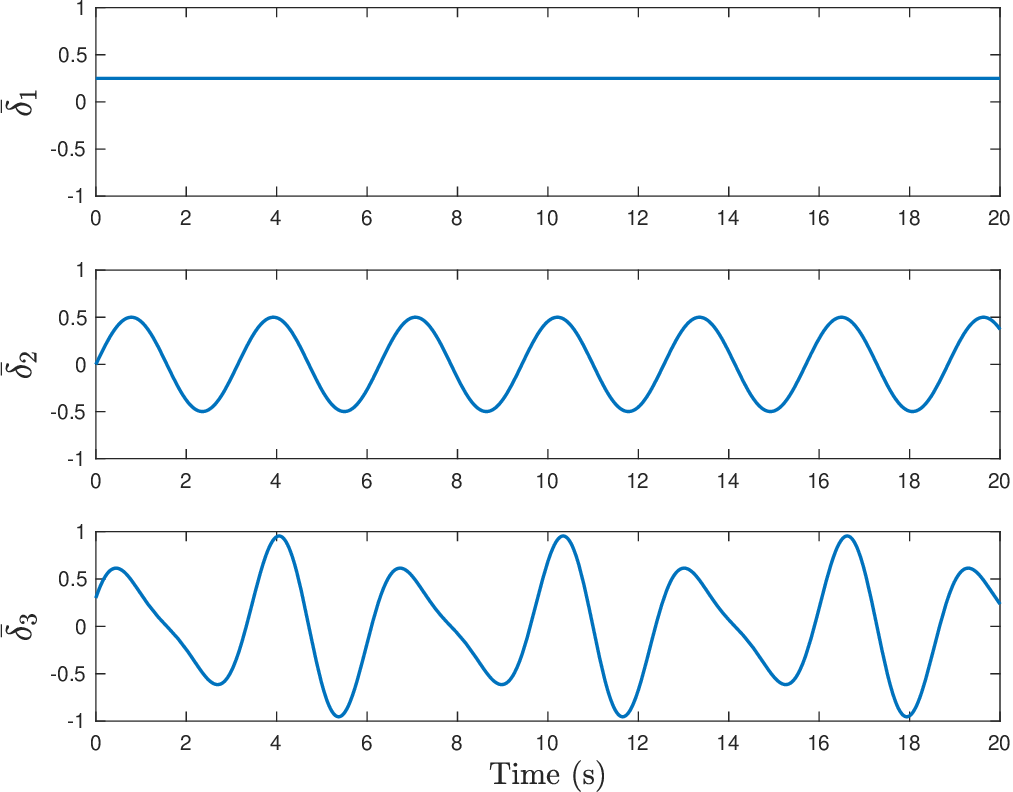}}
    \caption{The performance output and the convergence for various $\bar{\delta}_k$.}
    \label{F3}
\end{figure}
\begin{figure}[t!]
    \centering
    \includegraphics[width=0.4\linewidth]{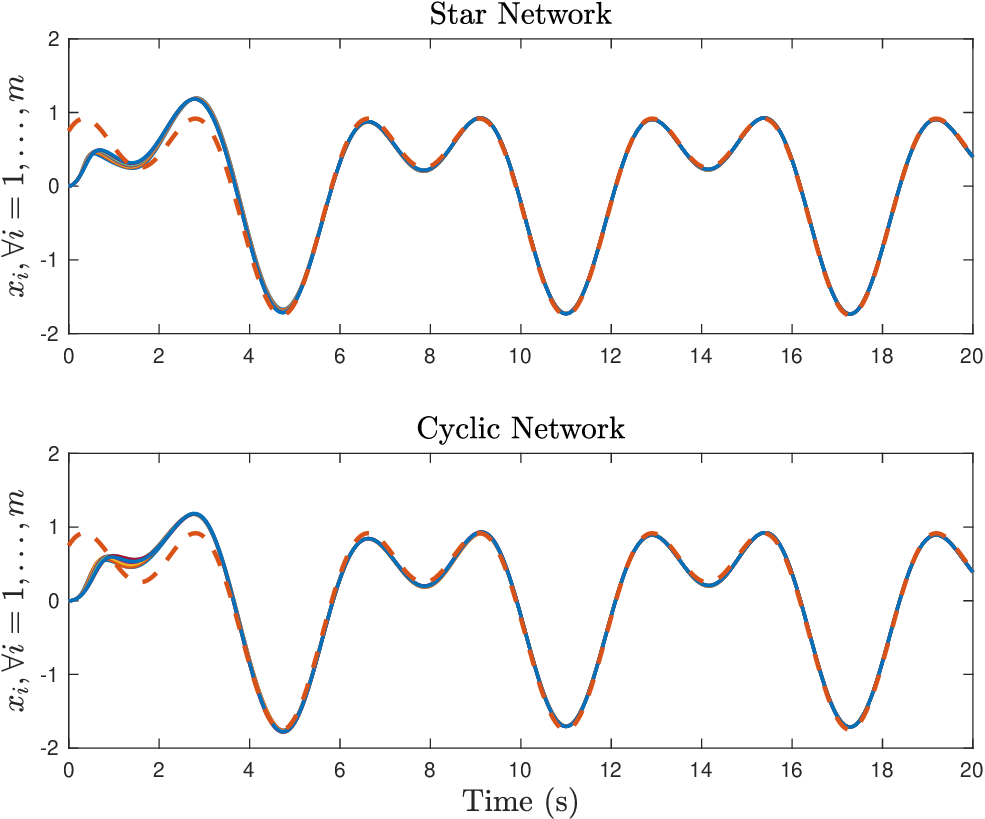}\qquad
    \includegraphics[width=0.4\linewidth]{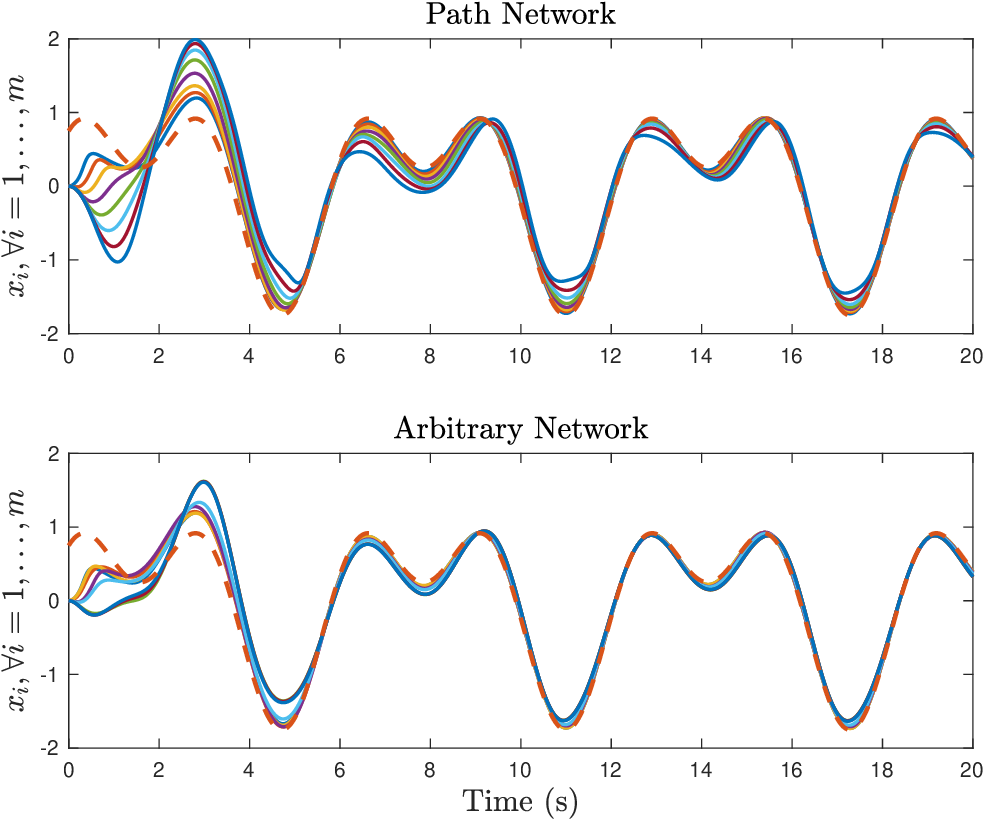}
    \caption{The performance outputs for four network topologies using $\bar{\delta}_3$}
    \label{F4}
\end{figure}
\begin{figure}[t!]
    \centering
    \subfloat[\label{F5a}]{\includegraphics[width=0.4\linewidth]{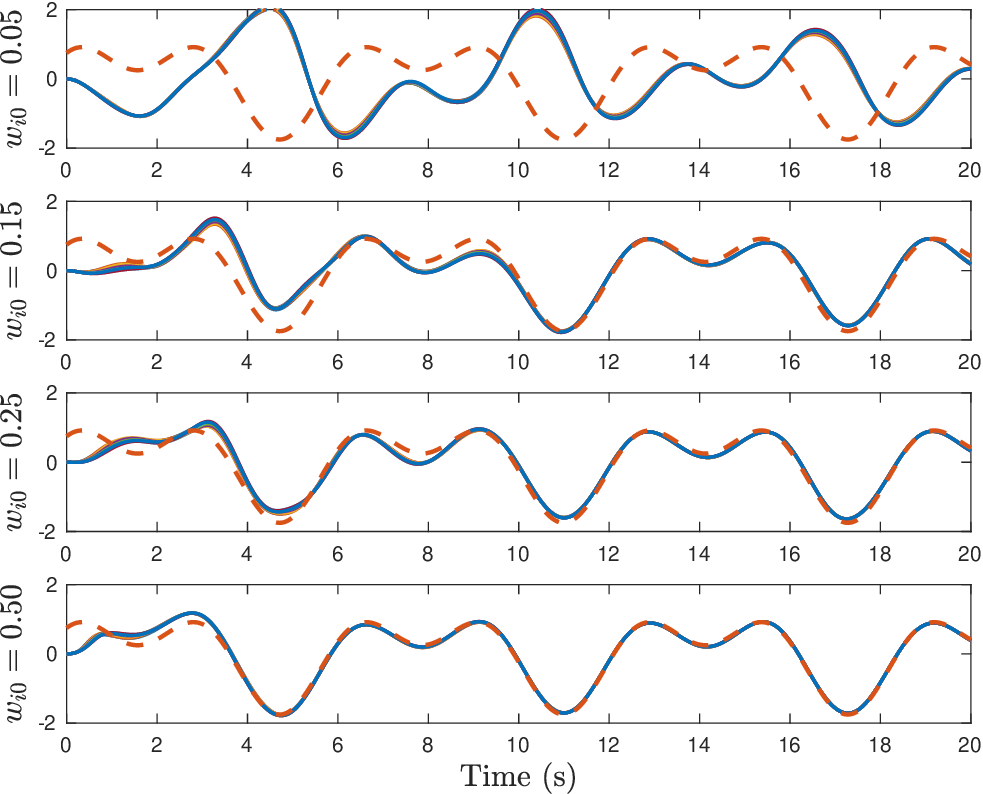}}\qquad
    \subfloat[\label{F5b}]{\includegraphics[width=0.4\linewidth]{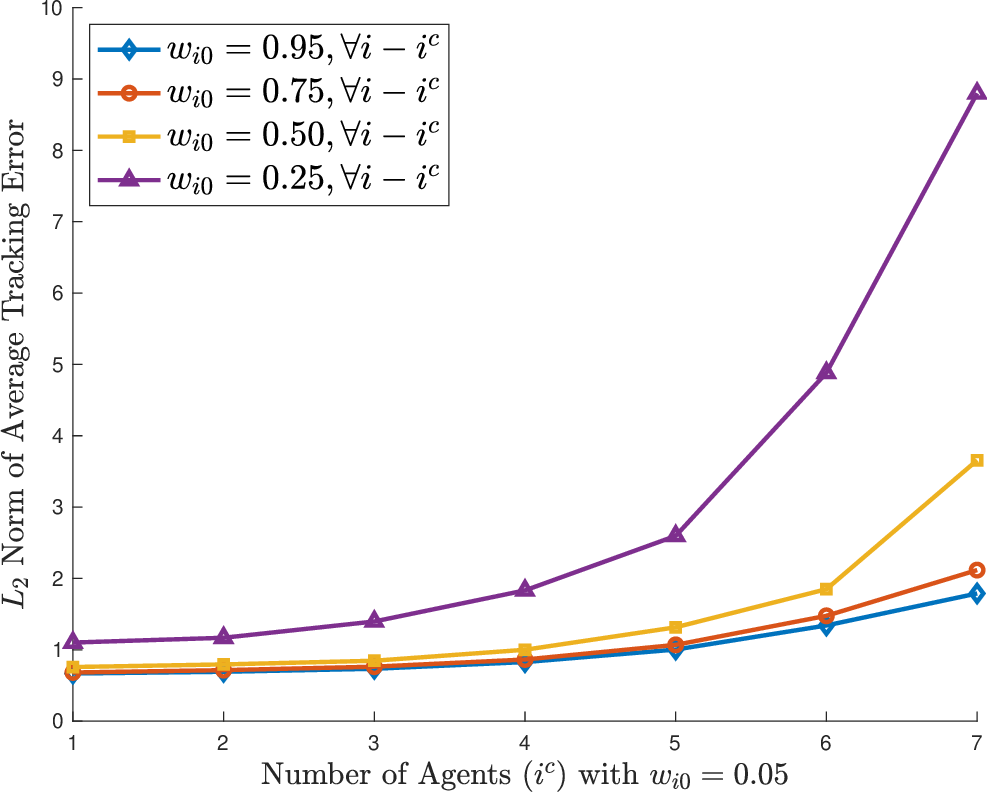}}
    \caption{Effect of leader weight $w_{i0}$ in cyclic networks under disturbance $\bar{\delta}_3$.}
    \label{F5}
\end{figure}
\begin{figure}[t!]
    \centering
    \subfloat[\label{F6a}]{\includegraphics[width=0.4\linewidth]{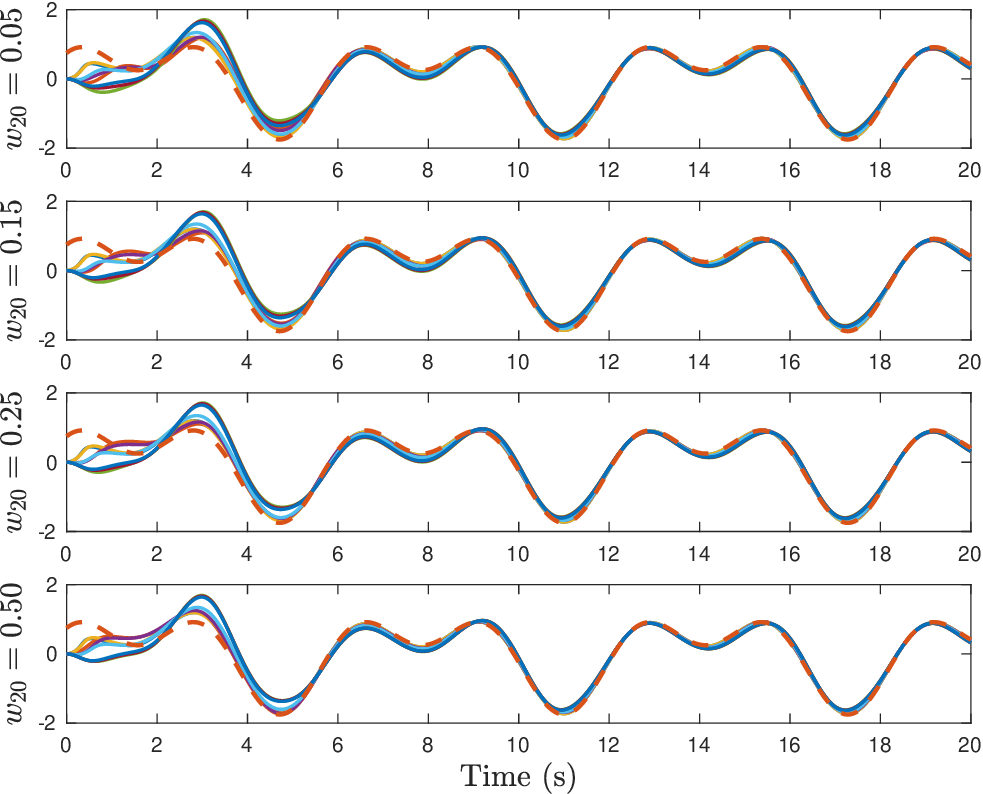}}\qquad
    \subfloat[\label{F6b}]{\includegraphics[width=0.4\linewidth]{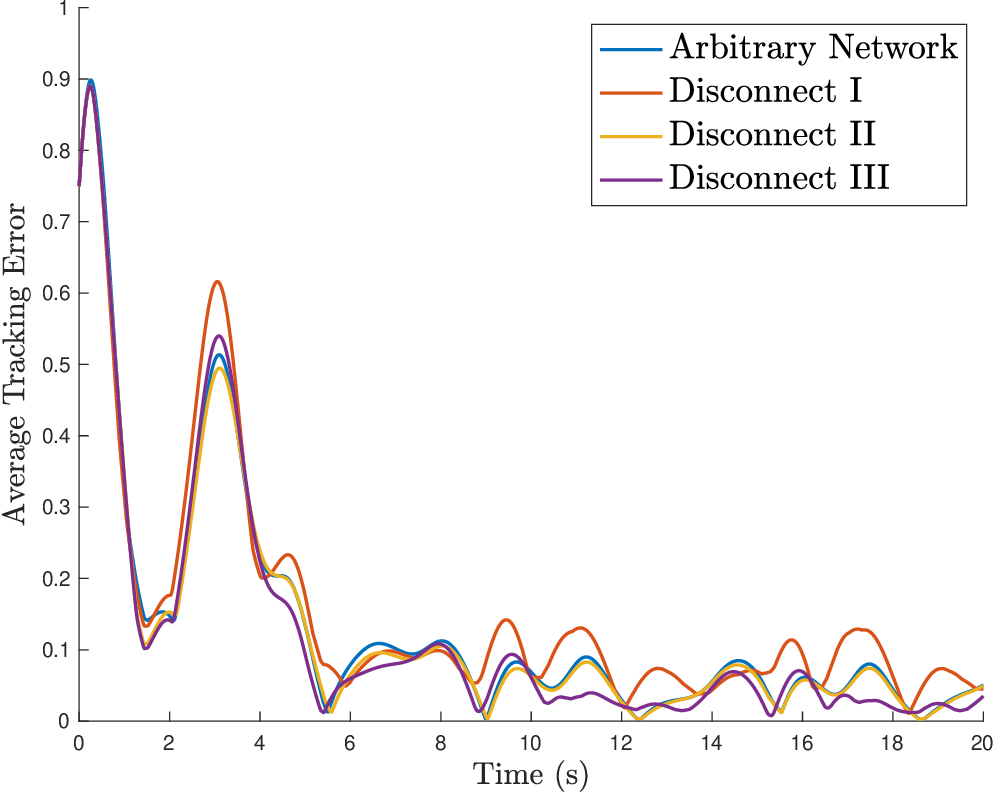}}
    \caption{Robustness of arbitrary network under varying leader and disconnected access.}
    \label{F6}
\end{figure}
\begin{figure}[t!]
    \centering
    \subfloat[\label{F7a}\scriptsize Scenario 1]{\includegraphics[width=0.5\linewidth]{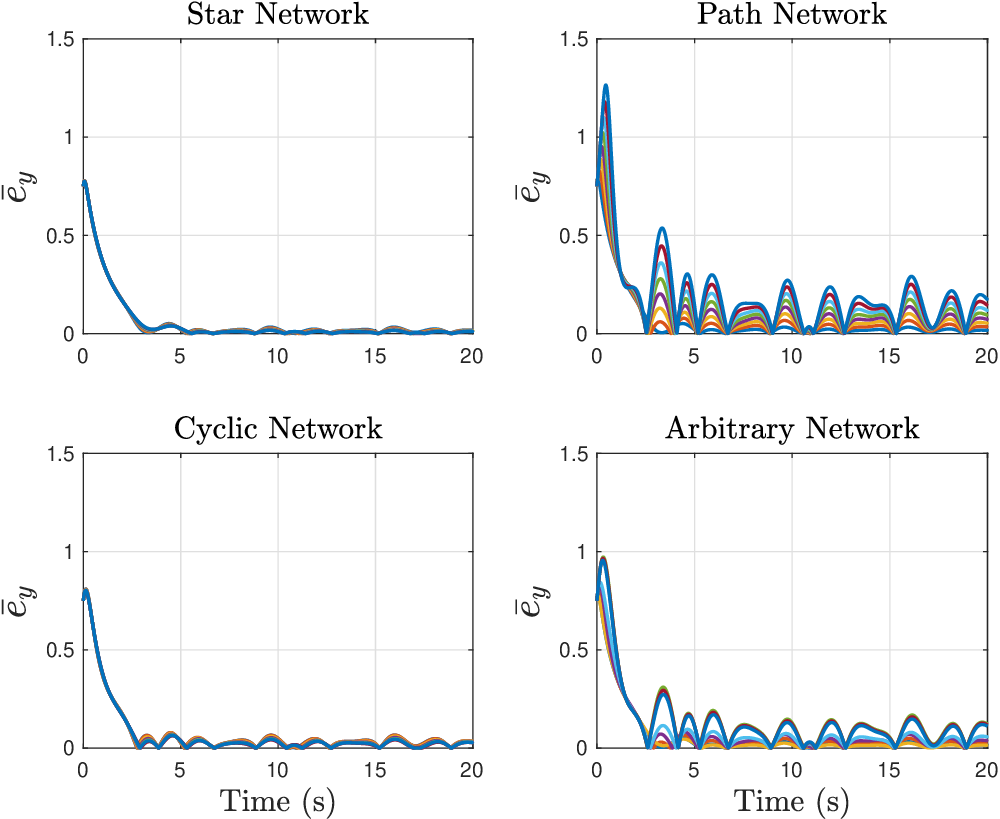}}
    \subfloat[\label{F7b}\scriptsize Scenario 2]{\includegraphics[width=0.5\linewidth]{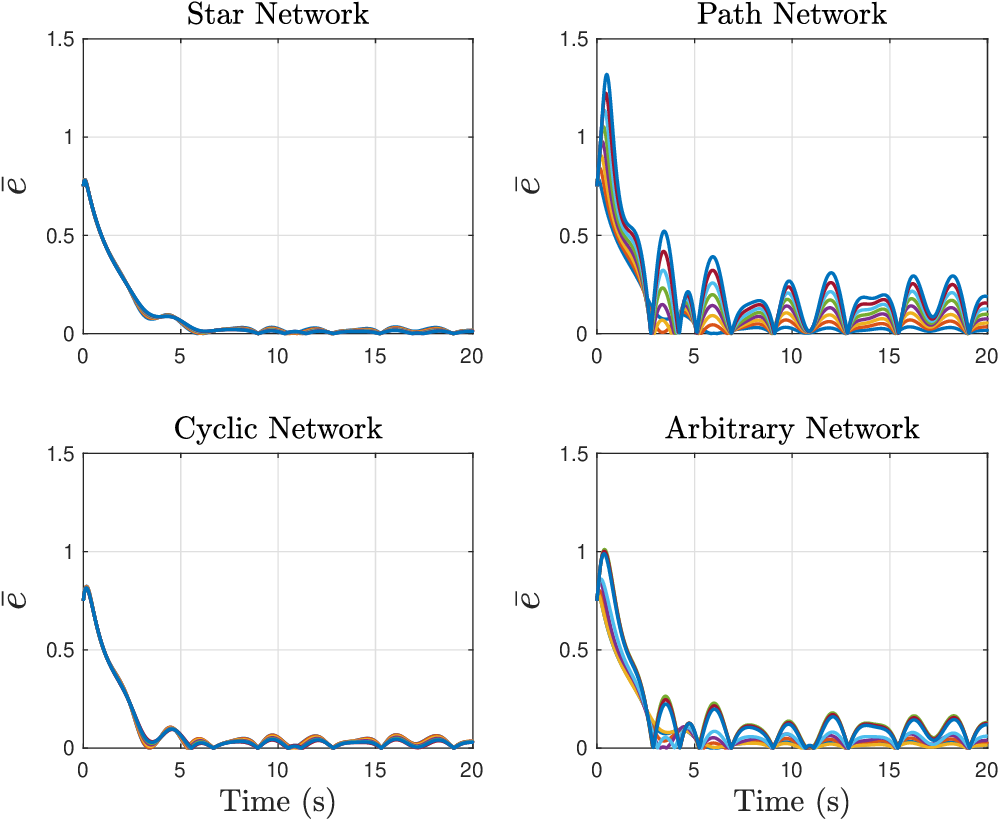}}
    \caption{Performance of distributed adaptive phase-lead compensator strategies.}
    \label{F7}
\end{figure}

\section{Conclusion}\label{sec:Conclusion}

This paper presented a unified framework for distributed adaptive synchronization under position-only communication. We developed two complementary designs: a passivity-certified (SPR) scheme via structured reparameterization, and a passivity-recovered (non-SPR) scheme that uses frequency shaping to restore SPR behavior. Both require no knowledge of the leader’s dynamics, operate on directed graphs satisfying balance and ``reachability'' conditions, and accommodate unknown (possibly unstable) follower dynamics.
Analysis shows global asymptotic synchronization in the disturbance-free case and under bounded time-varying disturbances; with parameter convergence under persistent excitation. Simulations on star, cyclic, path, and arbitrary topologies confirm robust tracking and scalability, and demonstrate that the frequency-shaped non-SPR designs match the SPR performance. Sensitivity experiments highlight how modest leader injection and alternative paths sustain performance despite structural variations. Future work will address communication delays and packet losses in the network \cite{Wafi-SIAM}.

\bibliographystyle{ieeetr}  
\bibliography{EN-Bib}  

\appendix
\begin{proof}[Derivation of \eqref{das:SPR:Edyn:dist}]\label{Apx:1}
With disturbance $\bar{\delta}$, the follower dynamics are \eqref{das:SPR:dynamics:dist}. Recall $\vartheta_u$, $\eta_u$, $\zeta_u$,
and the reparameterized control (cf. \eqref{das:SPR:mod:u})
\begin{equation*}
    \bar{u} = \Phi\vartheta_u + \hat{J}\zeta_u + \hat{B}\dot{\bar{x}}
    = \Phi\vartheta_u + \hat{\Xi}_u\eta_u,
\end{equation*}
where $\hat{\Xi}_u=[\hat{J},\hat{B}]$. Differentiating $\vartheta_u$ and using $\bar{e}=\bar{z}-\bar{x}$ as in \eqref{das:SPR:Edyn:u1}–\eqref{das:SPR:Edyn:u2} gives
\begin{equation}
    \dot{\vartheta}_u \;=\; \zeta_u - \ddot{\bar{x}}. \tag{A1} \label{A1}
\end{equation}
Solve \eqref{A1} for $\ddot{\bar{x}}$:
\begin{equation*}
    \ddot{\bar{x}} = J^{-1}\big(\bar{u} + \bar{\delta} - B\dot{\bar{x}}\big).
\end{equation*}
Substitute $\bar{u}$ and $\ddot{\bar{x}}$ into \eqref{A1}:
\begin{align*}
    \dot{\vartheta}_u
    &= \zeta_u - J^{-1}\big(\Phi\vartheta_u + \hat{J}\zeta_u + \hat{B}\dot{\bar{x}} + \bar{\delta} - B\dot{\bar{x}}\big) \\
    &= \big(I_m - \hat{J}J^{-1}\big)\zeta_u - J^{-1}\Phi\,\vartheta_u - J^{-1}(\hat{B}-B)\dot{\bar{x}} - J^{-1}\bar{\delta}.
\end{align*}
Using $\tilde{\Xi}_u \coloneqq \hat{\Xi}_u - \Xi_u = [\hat{J}-J, \hat{B}-B]$ and
$\tilde{\Xi}_u\eta_u = (\hat{J}-J)\zeta_u + (\hat{B}-B)\dot{\bar{x}}$, we group terms to obtain
\begin{equation*}
    \dot{\vartheta}_u = -J^{-1}\Phi\vartheta_u - J^{-1}\tilde{\Xi}_u\eta_u + J^{-1}\bar{\delta},
\end{equation*}
which is \eqref{das:SPR:Edyn:dist}.
\end{proof}

\end{document}